\PassOptionsToPackage{unicode}{hyperref}
\PassOptionsToPackage{hyphens}{url}
\PassOptionsToPackage{dvipsnames,svgnames,x11names}{xcolor}
\documentclass[
  12pt]{article}

\usepackage{algorithm}
\usepackage{algpseudocode}
\usepackage{setspace}
\usepackage{multirow}
\usepackage{xspace}

\usepackage{amsmath,amssymb}
\usepackage{amsthm}
\theoremstyle{plain}
\newtheorem{proposition}{Proposition}
\newtheorem{lemma}{Lemma}
\theoremstyle{remark}
\newtheorem{remark}{Remark}
\usepackage{iftex}
\ifPDFTeX
  \usepackage[T1]{fontenc}
  \usepackage[utf8]{inputenc}
  \usepackage{textcomp} % provide euro and other symbols
\else % if luatex or xetex
  \usepackage{unicode-math}
  \defaultfontfeatures{Scale=MatchLowercase}
  \defaultfontfeatures[\rmfamily]{Ligatures=TeX,Scale=1}
\fi
\usepackage{lmodern}
\ifPDFTeX\else
\fi
\IfFileExists{upquote.sty}{\usepackage{upquote}}{}
\IfFileExists{microtype.sty}{% use microtype if available
  \usepackage[]{microtype}
  \UseMicrotypeSet[protrusion]{basicmath} % disable protrusion for tt fonts
}{}
\makeatletter
\@ifundefined{KOMAClassName}{% if non-KOMA class
  \IfFileExists{parskip.sty}{%
    \usepackage{parskip}
  }{% else
    \setlength{\parindent}{0pt}
    \setlength{\parskip}{6pt plus 2pt minus 1pt}}
}{% if KOMA class
  \KOMAoptions{parskip=half}}
\makeatother
\usepackage{xcolor}
\usepackage{tikz}
\usetikzlibrary{calc,arrows.meta}
\usepackage{pgfplots}
\pgfplotsset{compat=1.18}
\newcommand{\swA}{\raisebox{2pt}{\rule{16pt}{2.2pt}}}
\newcommand{\swB}{\textcolor{black!45}{\raisebox{2pt}{\rule{6pt}{2.2pt}\hspace{2.5pt}\rule{6pt}{2.2pt}}}}
\newcommand{\swC}{\textcolor{black!45}{\raisebox{2pt}{\rule{2.2pt}{2.2pt}\hspace{2.5pt}\rule{2.2pt}{2.2pt}\hspace{2.5pt}\rule{2.2pt}{2.2pt}}}}
\makeatletter
\ifx\paragraph\undefined\else
  \let\oldparagraph\paragraph
  \renewcommand{\paragraph}{
    \@ifstar
      \xxxParagraphStar
      \xxxParagraphNoStar
  }
  \newcommand{\xxxParagraphStar}[1]{\oldparagraph*{#1}\mbox{}}
  \newcommand{\xxxParagraphNoStar}[1]{\oldparagraph{#1}\mbox{}}
\fi
\ifx\subparagraph\undefined\else
  \let\oldsubparagraph\subparagraph
  \renewcommand{\subparagraph}{
    \@ifstar
      \xxxSubParagraphStar
      \xxxSubParagraphNoStar
  }
  \newcommand{\xxxSubParagraphStar}[1]{\oldsubparagraph*{#1}\mbox{}}
  \newcommand{\xxxSubParagraphNoStar}[1]{\oldsubparagraph{#1}\mbox{}}
\fi
\makeatother

\usepackage{longtable,booktabs,array}
\usepackage{calc} % for calculating minipage widths
\usepackage{etoolbox}
\makeatletter
\patchcmd\longtable{\par}{\if@noskipsec\mbox{}\fi\par}{}{}
\makeatother
\IfFileExists{footnotehyper.sty}{\usepackage{footnotehyper}}{\usepackage{footnote}}
\makesavenoteenv{longtable}
\usepackage{graphicx}
\makeatletter
\def\maxwidth{\ifdim\Gin@nat@width>\linewidth\linewidth\else\Gin@nat@width\fi}
\def\maxheight{\ifdim\Gin@nat@height>\textheight\textheight\else\Gin@nat@height\fi}
\makeatother
\setkeys{Gin}{width=\maxwidth,height=\maxheight,keepaspectratio}
\makeatletter
\def\fps@figure{htbp}
\makeatother

\makeatletter
\@ifpackageloaded{caption}{}{\usepackage{caption}}
\AtBeginDocument{%
\ifdefined\contentsname
  \renewcommand*\contentsname{Table of contents}
\else
  \newcommand\contentsname{Table of contents}
\fi
\ifdefined\listfigurename
  \renewcommand*\listfigurename{List of Figures}
\else
  \newcommand\listfigurename{List of Figures}
\fi
\ifdefined\listtablename
  \renewcommand*\listtablename{List of Tables}
\else
  \newcommand\listtablename{List of Tables}
\fi
\ifdefined\figurename
  \renewcommand*\figurename{Figure}
\else
  \newcommand\figurename{Figure}
\fi
\ifdefined\tablename
  \renewcommand*\tablename{Table}
\else
  \newcommand\tablename{Table}
\fi
}
\@ifpackageloaded{float}{}{\usepackage{float}}
\floatstyle{ruled}
\@ifundefined{c@chapter}{\newfloat{codelisting}{h}{lop}}{\newfloat{codelisting}{h}{lop}[chapter]}
\floatname{codelisting}{Listing}

\makeatother
\makeatletter
\@ifpackageloaded{caption}{}{\usepackage{caption}}
\@ifpackageloaded{subcaption}{}{\usepackage{subcaption}}
\makeatother

\ifLuaTeX
  \usepackage{selnolig}  % disable illegal ligatures
\fi
\usepackage[]{natbib}
\usepackage{bookmark}

\IfFileExists{xurl.sty}{\usepackage{xurl}}{} % add URL line breaks if available
\newcommand{\papertitle}{A Geometry-Aware Framework for Clustering Cylindrical Data}
\hypersetup{
  pdftitle={\papertitle},
  pdfauthor={Giuseppe Pandolfo; Luca Coraggio; Antonio D'Ambrosio},
  pdfkeywords={Circular-linear data; Color quantization; Directional statistics; Fr\'echet mean; Geodesic distance; K-means algorithm},
  colorlinks=true,
  linkcolor={blue},
  filecolor={Maroon},
  citecolor={Blue},
  urlcolor={Blue},
  pdfcreator={LaTeX via pandoc}}

\newcommand{\anon}{1}

\if0\anon\hypersetup{pdfauthor={}}\fi

\newcommand{\euckmeans}{\texttt{Euclidean k-means++}\xspace}
\newcommand{\chordkmeans}{\texttt{Chordal k-means++}\xspace}
\newcommand{\geokmeans}{\texttt{Geodesic k-means++}\xspace}
\newcommand{\sengupta}{\texttt{SenGupta}\xspace}
\newcommand{\rgbkmeans}{\texttt{RGB k-means++}\xspace}
\newcommand{\nox}{NO\textsubscript{x}\xspace}
\newcommand{\ie}{\textit{i.e.}\xspace}
\newcommand{\eg}{\textit{e.g.}\xspace}

\begin{document}

\def\spacingset#1{\renewcommand{\baselinestretch}%
  {#1}\small\normalsize} \spacingset{1}

%%%%%%%%%%%%%%%%%%%%%%%%%%%%%%%%%%%%%%%%%%%%%%%%%%%%%%%%%%%%%%%%%%%%%%%%%%%%%%

\if1\anon
  {
    \title{\bf \papertitle}
    \author{Giuseppe Pandolfo$^{1}$\thanks{Corresponding author.},
      Luca Coraggio$^{1,2}$
      and Antonio D'Ambrosio$^{1}$\\[1.5ex]
      {\normalsize $^{1}$Dipartimento di Scienze Economiche e Statistiche,}\\
      {\normalsize Universit\`a degli Studi di Napoli Federico II, 80126 Naples, Italy}\\[0.5ex]
      {\normalsize $^{2}$Centre for Studies in Economics and Finance (CSEF),}\\
      {\normalsize University of Naples Federico II, 80126 Naples, Italy}\\[1.5ex]
      {\small \texttt{giuseppe.pandolfo@unina.it}, \texttt{luca.coraggio@unina.it}, \texttt{antdambr@unina.it}}}
    \maketitle
  } \fi

\if0\anon
  {
    \bigskip
    \bigskip
    \bigskip
    \begin{center}
      {\LARGE\bf \papertitle}
    \end{center}
    \medskip
  } \fi

\bigskip
\begin{abstract}
  Cylindrical data pair an angle with a linear measurement. Clustering that ignores the periodicity of the angle breaks up groups lying across its origin. We formulate the \(K\)-means algorithm for a generic distance on the cylinder and instantiate it with the chordal distance of the ambient space and the geodesic distance along the surface, so that the two versions differ in the metric alone. Centroids are exact Fr\'echet means computed with known tools: the mean direction for the chord, the circular Fr\'echet mean for the arc. Seeded by \(K\)-means++, the algorithm converges in finitely many iterations at a cost comparable to classical \(K\)-means. Over an extensive empirical analysis on simulated data, Euclidean \(K\)-means is never meaningfully better and breaks down when clusters cross the origin; the two versions agree for concentrated angles, the chord incidentally prevailing only when the angle is uninformative; a model-based cylindrical mixture is the most variable method but describes elongated, correlated clusters better. In hue--value color image quantization, the cylindrical versions performs better on object detection based on color segmentation; on wind direction and nitrogen oxides data they recover two pollution regimes that Euclidean \(K\)-means cuts apart.
\end{abstract}

\noindent%
\vfill

\newpage
\spacingset{1.8} % DON'T change the spacing!

\section{Introduction}\label{sec:intro}

Clustering is a fundamental subject in statistics that has been widely investigated in classical multivariate analysis over the past few decades. The goal is to organize a set of objects into homogeneous groups such that the objects within the same group are more similar to each other than to objects in different groups. Due to its practical importance, clustering has received significant attention in various branches of statistics.

Less attention has been paid to cluster analysis of cylindrical data even though in modern applications we often encounter data consisting of both linear and circular variables. In other words, we observe directional data. Few examples are wind direction and another climatological variable such as wind speed or air temperature, the direction an animal moves and the distance moved, or wave direction and wave height. Recently, analyses of cylindrical data concentrated on the exploration of wind direction and SO2 concentration \citep{garcia2013exploring}, the analysis of Japanese earthquakes \citep{wang2013analysis}, the link between wildfire orientation and burnt area \citep{garcia2014test}, and space-time modeling of sea currents in the Adriatic Sea \citep{wang2015joint, lagona2015hidden}.

Cylindrical data are constrained to lie on the lateral surface of a cylinder $\mathcal{C} = S^1 \times \mathbb{R} = \{ (\theta, z) : \theta \in [0, 2\pi), \ z \in \mathbb{R} \}$, where $S^1 = \{ e^{i\theta} : \theta \in [0, 2\pi) \}$.

Analyzing such data must respect the manifold structure underlying the data. Traditional clustering algorithms neglect the cylindrical topology, and often produce misleading results due to the periodicity of the angular component and the combined linear-circular interactions.

Despite the growing interest, within the literature just a few of model-based proposals for clustering cylindrical data can be found. \cite{sengupta2020model} introduced
a mixture model approach based on the joint distribution of the linear and the circular variable for model based clustering of such data. \cite{fujita2017clustering} proposed a clustering method for data in cylindrical coordinates based on the $k$-means by deriving a  similarity measure from the log likelihood of the assumed probability distribution (i.e. the von Mises and Gaussian for the circular and linear variable, respectively). Then, to the best of authors' knowledge, there is not any other work aimed at providing techniques for clustering cylindrical data.

In this work, we propose a geometry-aware framework for clustering cylindrical data which is distance-based, and therefore independent of the distribution of the given data.
The framework is the \(K\)-means algorithm written for a generic distance on the cylinder, with centroids defined as Fr\'echet means, and we instantiate it with two distances that respect the periodicity of the angle: the chordal distance, inherited from the ambient space in which the cylinder is embedded, and the geodesic distance, defined as the length of the shortest path between two points along the surface.
The ingredients are classical. The chord leads to the mean direction of directional statistics \citep{MardiaJupp2000}, while the angular part of the geodesic distance, the arc length on the circle, and the Fr\'echet mean it induces have been studied by \citet{McKilliam2012}, \citet{HotzHuckemann2015} and \citet{Cazals2021}; the last of these applies it within a \(k\)-means++ algorithm on the flat torus.
Our contribution is to bring them into a single algorithm for cylindrical data, in which the two variants share seeding, assignment rule and stopping criterion and differ in the metric alone, to establish its convergence and its computational cost, and to assess, through simulations and two applications, when accounting for the geometry matters and which of the two metrics is to be preferred.
Such approach accounts for both circular and linear components simultaneously while preserving their geometric constraints, and aims to be considered as an alternative to the model-based tools available for clustering cylindrical data.

The paper is organized as follows. Section~\ref{sec:method} introduces the two distances on the cylinder, derives the geodesic one as the length of the helix joining two points on the lateral surface (Section~\ref{subsec:geodesic_distance}), and states the clustering objective and the resulting algorithm for a generic distance, together with its computational cost (Section~\ref{subsec:method}) and its convergence guarantee (Section~\ref{subsec:convergence}). Section~\ref{sec:simulations} evaluates the two variants against the Euclidean \(K\)-means and a model-based mixture through a simulation study, and Sections~\ref{sec:images} and~\ref{sec:case1} apply them to color image quantization and to wind direction and pollutant concentration data.

\section{Method}\label{sec:method}
%\section{The geodesic distance on the cylinder}\label{}
%\section{Geodesic Distance on the Lateral Surface of the Unit Cylinder}

Since clustering is the grouping of similar objects, a measure that can determine whether two cylindrical objects are similar or dissimilar is required. Two such measures are considered here: the Euclidean distance of the ambient space, and the geodesic distance along the lateral surface.

Let \( \mathcal{M} \) denote the lateral surface of the cylinder of radius \(r > 0\) embedded in \(\mathbb{R}^3\), defined as
\[
\mathcal{M} = \{ (r\cos \theta, r\sin \theta, z) \in \mathbb{R}^3 : \theta \in [0, 2\pi), z \in \mathbb{R} \}.
\]
Whatever its radius, this surface can be viewed as the Cartesian product of the circle \( S^1 \) and the real line \( \mathbb{R} \), i.e., \( \mathcal{M} \cong S^1 \times \mathbb{R} = \mathcal{C} \) where each observation $x \in \mathcal{M}$ can be represented as
\[
x = (\theta, z), \quad \theta \in [0, 2\pi), \ z \in \mathbb{R}.
\]
The radius does not enter the coordinates of a point but only the distances between points, where it sets how much a difference in angle weighs against a difference in \(z\).

The Euclidean distance between two points \(x = (\theta_x, z_x)\) and \(y = (\theta_y, z_y)\) of \(\mathcal{M}\), taken in the ambient space \(\mathbb{R}^3\), measures their angular separation by the \emph{chord} that subtends it in the plane, that is, by \(r\) times the chord of the unit circle,
\begin{equation}
\delta^{\mathrm{chord}}(x,y) = 2\Big|\sin\tfrac{\theta_x - \theta_y}{2}\Big| ,
\label{eq:deltachord}
\end{equation}
and reads
\begin{equation}
d_{\mathrm{chord}}(x, y)^2 = r^2\,\delta^{\mathrm{chord}}(x,y)^2 + (z_x - z_y)^2 = 2r^2\big(1 - \cos(\theta_x - \theta_y)\big) + (z_x - z_y)^2 .
\label{eq:dchord}
\end{equation}
We refer to \(d_{\mathrm{chord}}\) as the \emph{chordal} distance. It respects the periodicity of \(\theta\), and the angular centroid it induces is available in closed form: the sum of squared chords from a set of angles \(\theta_1, \ldots, \theta_m\) is minimized by their mean direction \(\mathrm{atan2}\big(\sum_i \sin\theta_i, \sum_i \cos\theta_i\big)\) \citep{MardiaJupp2000}. The chord, however, cuts through the interior of the cylinder rather than following its surface, and minimizing squared chords is not the same as minimizing squared arcs: the two circular means that result, the mean direction and the Fr\'echet mean of the arc length, are compared by \citet{McKilliam2012} and \citet{HotzHuckemann2015}. The distance that follows the surface is derived next.

\subsection{The geodesic distance}\label{subsec:geodesic_distance}

The shortest path (geodesic) on the lateral surface of the cylinder between two points \(x = (\theta_x, z_x)\) and \(y = (\theta_y, z_y)\) is an arc of a cylindrical \textbf{helix}, as depicted in Figure~\ref{fig:fig1}. Its angular extent is the minimal angular difference
\[
\delta^{\mathrm{geo}}(x, y) = \min_{\varphi \in \mathbb{Z}} | \theta_x - \theta_y + 2 \pi \varphi | ,
\]
where \(\mathbb{Z}\) denotes the set of all integers and \(\varphi\) is chosen so as to minimize the angular distance on the unit circle \(S^1\). Since \(\theta_x, \theta_y \in [0, 2\pi)\), the difference \(\theta_x - \theta_y\) lies in \((-2\pi, 2\pi)\), so the minimizing \(\varphi\) necessarily belongs to \(\{-1, 0, 1\}\), and the minimal angular difference can accordingly be computed as
\[
\delta^{\mathrm{geo}}(x, y) = \min \big( |\theta_x - \theta_y|, |\theta_x - \theta_y + 2\pi|, |\theta_x - \theta_y - 2\pi| \big).
\]
This ensures the geodesic distance correctly accounts for the circular periodicity. Parameterized by the angle \(u\) traveled along the shorter arc, the helix reads
\[
\boldsymbol{\gamma}(u) = \big( r\cos(\theta_x + s u), \ r\sin(\theta_x + s u), \ z_x + h u \big), \quad u \in [0, \delta^{\mathrm{geo}}(x, y)],
\]
where \(s \in \{-1, +1\}\) is the direction of the shorter arc from \(\theta_x\) to \(\theta_y\) and \(h \in \mathbb{R}\) is the helix pitch, the vertical displacement per unit angle,
\[
h = \frac{z_y - z_x}{\delta^{\mathrm{geo}}(x, y)} .
\]
Hence, the length of this helix segment is
\begin{equation}\label{eq:dgeo}
\begin{aligned}
d_{\mathrm{geo}}(x, y) = L &= \int_{0}^{\delta^{\mathrm{geo}}(x,y)} \|\boldsymbol{\gamma}'(u)\| \, du = \sqrt{r^2 + h^2} \cdot \delta^{\mathrm{geo}}(x,y) \\
&= \sqrt{r^2\,\delta^{\mathrm{geo}}(x,y)^2 + (z_x - z_y)^2}.
\end{aligned}
\end{equation}
This is the exact formula for the geodesic distance on the cylinder, which we henceforth denote by \(d_{\mathrm{geo}}\). Both angular measures are functions of the same separation: the chord~\eqref{eq:deltachord} is \(\delta^{\mathrm{chord}}(x,y) = 2\sin\big(\delta^{\mathrm{geo}}(x,y)/2\big)\), since \(\delta^{\mathrm{geo}}(x,y) \in [0,\pi]\).

When \(h = 0\), the helix reduces to a circular arc of length $r\,\delta^{\mathrm{geo}}(x,y)$. For \(h \neq 0\), the curve ascends vertically while wrapping around the cylinder, increasing the overall length. With $\theta_{x}=\theta_{y}$ the geodesic path is simply a vertical line segment.

\begin{figure}[htbp]
  \centering
  \includegraphics[width=0.7\textwidth]{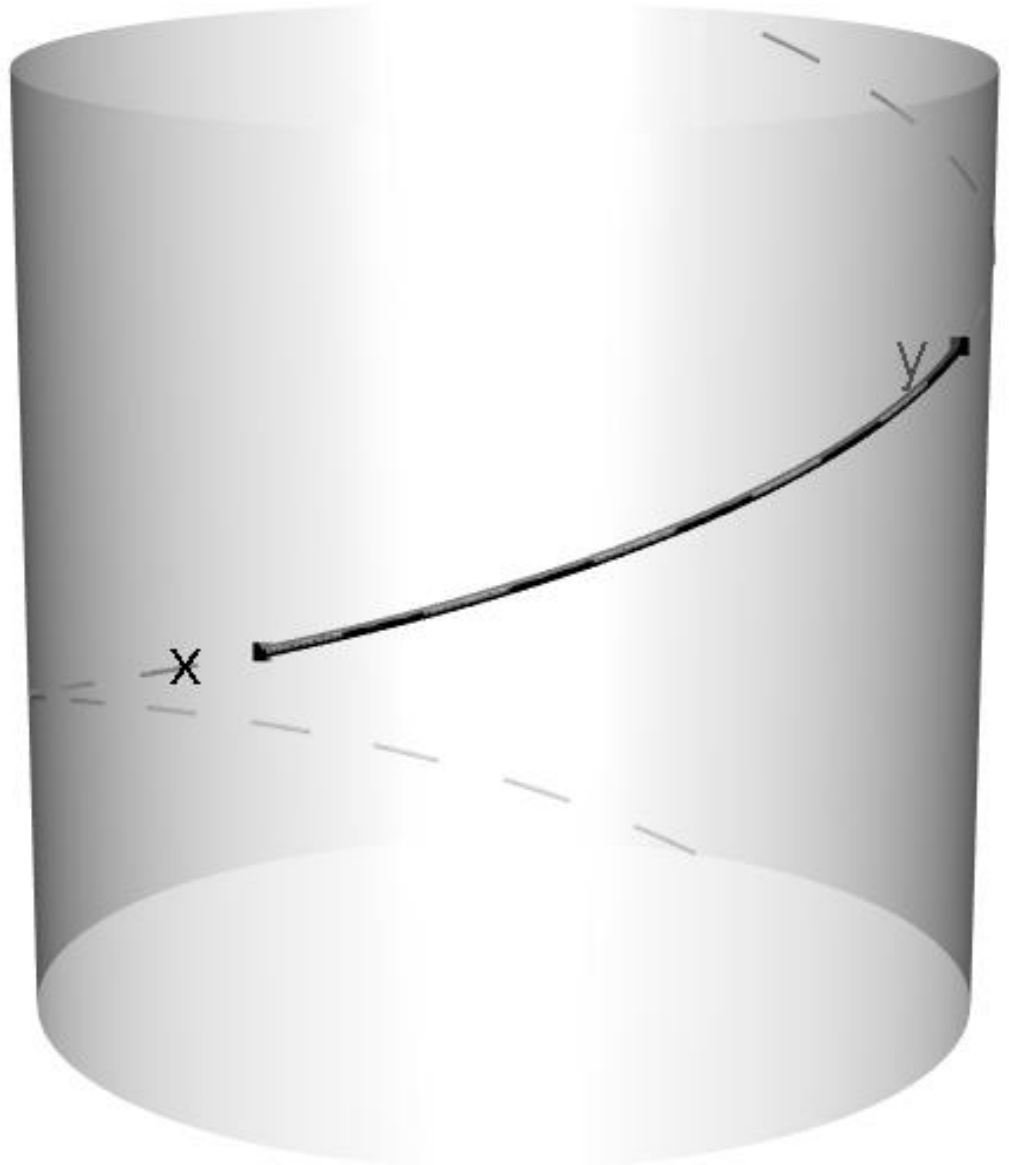}
  \caption{The geodesic distance between two points of the cylinder is the length of the helix joining them.}\label{fig:fig1}
\end{figure}

One can note that it naturally generalizes Euclidean distance by incorporating circular periodicity in \(\theta\) which cannot be neglected when analyzing cylindrical data. %To illustrate the difference between two measures, the corresponding contour plot

%---------------------------------------------------
%The geodesic distance \( d_{\mathcal{M}} \) between any two points \( x = (\theta_x, z_x) \) and \( y = (\theta_y, z_y) \) on \( \mathcal{M} \), where we identify points via their cylindrical coordinates, is the length of the shortest path on the surface connecting them. Since the cylinder is a flat surface with no curvature along its axis and periodic in the angular direction, the geodesic distance is given by
%\begin{equation}
%d_{\mathcal{M}}(x, y) = \sqrt{\delta^{\mathrm{geo}}(x, y)^2 + (z_x - z_y)^2},
%\label{eq:geodesic_cylinder}
%\end{equation}
%where
%\[
%\delta^{\mathrm{geo}}(x, y) = \min_{\varphi \in \mathbb{Z}} | \theta_x - \theta_y + 2 \pi \varphi |
%\]
%is the minimal angular difference. Where $\mathbb{Z}$ denotes he set of all integers and $\varphi$ is chosen to minimize the angular distance on the unit circle \( S^1 \). In practice, it suffices to check only a small set of values of \( k \), typically
%\[
%\varphi \in \{-1, 0, 1\},
%\]
%since adding or subtracting one full rotation \( 2\pi \) once is enough to find the shortest arc length. Therefore, the minimal angular difference can be computed as
%\[
%\delta^{\mathrm{geo}}(x, y) = \min \big( |\theta_x - \theta_y|, |\theta_x - \theta_y + 2\pi|, |\theta_x - \theta_y - 2\pi| \big).
%\]
%This ensures the geodesic distance correctly accounts for the circular periodicity.
%---------------------------------------------------

The geodesic distance \( d_{\mathrm{geo}} \) satisfies the following properties:

\begin{itemize}
  \item \textbf{Non-negativity:} $d_{\mathrm{geo}}(x, y) \geq 0$ for all $x, y \in \mathcal{C}$, and $d_{\mathrm{geo}}(x, y) = 0 \Leftrightarrow x = y$.

  \item \textbf{Symmetry:} \( d_{\mathrm{geo}}(x, y) = d_{\mathrm{geo}}(y, x) \).

  \item \textbf{Triangle inequality:} For any \( x, y, w \in \mathcal{C} \),
        \[
        d_{\mathrm{geo}}(x, w) \leq d_{\mathrm{geo}}(x, y) + d_{\mathrm{geo}}(y, w).
        \]
\end{itemize}

\subsection{\texorpdfstring{\(K\)-means clustering for cylindrical data}{K-means clustering for cylindrical data}}\label{subsec:method}

Consider a dataset \( \{x_i\}_{i=1}^n \) of size $n$ where each observation \( x_i = (\theta_i, z_i) \) lies on the cylindrical manifold $\mathcal{C} = S^{1} \times \mathbb{R}$, with \( \theta_i \in [0, 2\pi) \) representing the angular component and \( z_i \in \mathbb{R} \) the linear component, and let \(d \in \{d_{\mathrm{geo}}, d_{\mathrm{chord}}\}\) denote either of the two distances introduced above, for a fixed radius \(r > 0\).

Given a predefined number of clusters \( K (>1) \), the goal is to partition the dataset into clusters \( \{ C_k \}_{k=1}^K \) by minimizing the within-cluster sum of squared distances:
\begin{equation}
\min_{\{ C_k \}_{k=1}^K} \sum_{k=1}^K \sum_{x_i \in C_k} d \left( x_i, \mu_k \right)^2,
\label{eq:clustering_objective}
\end{equation}
where \(\mu_k \in \mathcal{C}\) represents the cluster centroid computed as the Fr\'echet mean:
\[
\mu_k = \arg\min_{\mu \in \mathcal{C}} \sum_{x_i \in C_k} d(x_i, \mu)^2.
\]

Here and in what follows we write, with a slight abuse of notation, \(\delta^{\mathrm{geo}}(\theta_x, \theta_y) = \min_{\varphi \in \mathbb{Z}} |\theta_x - \theta_y + 2\pi\varphi|\) for the minimal difference between two angles, and likewise \(\delta^{\mathrm{chord}}(\theta_x, \theta_y)\). Both squared distances, \eqref{eq:dgeo} and~\eqref{eq:dchord}, then separate additively into an angular and a linear term, \(d(x, y)^2 = r^2\,\delta(\theta_x, \theta_y)^2 + (z_x - z_y)^2\) with \(\delta \in \{\delta^{\mathrm{geo}}, \delta^{\mathrm{chord}}\}\) the angular distance corresponding to \(d\), so the Fr\'echet mean factorizes into two independent one-dimensional problems, one per coordinate. The factor \(r^2\) multiplies the whole angular problem and leaves its minimizer unchanged: the radius affects which centroid a point is assigned to, not where the centroid of a given cluster lies. Under either distance the linear component of the centroid is the arithmetic mean
\begin{equation}
\bar{z}_k = \frac{1}{|C_k|} \sum_{x_i \in C_k} z_i,
\label{eq:centroid_linear}
\end{equation}
where \( |C_k| \) is the size of the $k$-th cluster \( C_k \). The angular component is the Fr\'echet mean on the circle under the corresponding angular distance:
\begin{itemize}
  \item under \(d_{\mathrm{geo}}\), the minimizer of the sum of squared arc lengths,
        \begin{equation}
        \bar{\theta}_k = \arg\min_{\vartheta \in [0, 2\pi)} \, \sum_{x_i \in C_k} \delta^{\mathrm{geo}}(\theta_i, \vartheta)^2 ,
        \label{eq:centroid_geo}
        \end{equation}
        which admits no closed form but can be computed exactly, in \(O(|C_k| \log |C_k|)\) operations \citep{McKilliam2012, HotzHuckemann2015, Cazals2021}; Appendix~\ref{sec:frechet} gives a self-contained derivation and the procedure used here, Algorithm~\ref{alg:frechet_circular};
  \item under \(d_{\mathrm{chord}}\), the minimizer of the sum of squared chords \(\sum_{x_i \in C_k} 2\big(1 - \cos(\theta_i - \vartheta)\big)\), namely the mean direction \citep{MardiaJupp2000}, is 
        \begin{equation}
        \bar{\theta}_k = \mathrm{atan2}\Big( \sum_{x_i \in C_k} \sin\theta_i, \ \sum_{x_i \in C_k} \cos\theta_i \Big).
        \label{eq:centroid_meandir}
        \end{equation}
\end{itemize}
The Fr\'echet mean always exists: the angular objective is continuous on the compact circle \(S^1\), while the linear objective is a strictly convex quadratic in \(\bar{z}\); a minimizer of their sum therefore exists on \(\mathcal{C} = S^1 \times \mathbb{R}\). 
The arc-length objective~\eqref{eq:centroid_geo} may have several minimizers (\eg for \(m\) angles equispaced on the circle it is minimized at \(m\) distinct points), while the minimizer of the sum of squared chords is unique whenever the resultant length \(R_k = \big\| \sum_{x_i \in C_k} (\cos\theta_i, \sin\theta_i) \big\|\) is positive, and every \(\vartheta\) is a minimizer when \(R_k = 0\), the sum being then constant \citep{MardiaJupp2000}. 

The clustering algorithm iteratively performs the following steps until convergence:
\begin{enumerate}
  \item \textbf{Assignment step:} Assign each point \( x_i \) to the cluster whose centroid is nearest under \(d\):
        \[
        C_k = \left\{ x_i : k = \arg\min_{j} d(x_i, \mu_j) \right\}.
        \]
  \item \textbf{Update step:} Recompute the centroids \( \mu_k \) for each cluster according to~\eqref{eq:centroid_linear} and to~\eqref{eq:centroid_geo} or~\eqref{eq:centroid_meandir}.
\end{enumerate}
This is the classical \( K \)-means algorithm on the cylindrical manifold, with \(d\) in place of the Euclidean distance.

Initialization of iterative algorithms can have a significant impact on the resulting performances. Here, the initialization adapts the classical \(K\)-means++ method of \cite{arthur2006k} to the cylinder: the first centroid is selected uniformly at random from the dataset, and each subsequent centroid is a data point chosen with probability proportional to its squared distance from the nearest already chosen centroid. Formally, if \(\{\mu_1, \ldots, \mu_{k-1}\}\) are the previously selected centroids, the probability of selecting \(x_i\) as the next centroid is
\[
P(x_i) = \frac{D_i}{\sum_{j=1}^n D_j}, \quad \text{where } D_i = \min_{1 \leq j < k} d(x_i, \mu_j)^2,
\]
so that the initial centroids are well spread with respect to the distance in use. The whole procedure is summarized in Algorithm~\ref{alg:geodesic_kmeans}. We refer to it as the \emph{geodesic} \(K\)-means when \(d = d_{\mathrm{geo}}\) and as the \emph{chordal} \(K\)-means when \(d = d_{\mathrm{chord}}\); seeding, assignment rule and stopping criterion are identical in the two variants, so that their comparison isolates the effect of the metric alone.

\begin{algorithm}[htbp]
  \begin{spacing}{0.8}
    \caption{\(K\)-means clustering for cylindrical data}
    \label{alg:geodesic_kmeans}
    \begin{algorithmic}[1]
      \Require Dataset \( \{x_i = (\theta_i, z_i)\}_{i=1}^n \subset \mathcal{C} = S^1 \times \mathbb{R} \), radius \( r > 0 \), distance \( d \in \{d_{\mathrm{geo}}, d_{\mathrm{chord}}\} \) with angular component \( \delta \in \{\delta^{\mathrm{geo}}, \delta^{\mathrm{chord}}\} \), number of clusters \( K \), maximum iterations \( T \)
      \Ensure Cluster assignments \( \{C_k\}_{k=1}^K \), centroids \( \{\mu_k = (\bar{\theta}_k, \bar{z}_k)\}_{k=1}^K \)

      \State \textbf{Initialization: $K$-means++}
      \State Randomly select the first centroid \(\mu_1^{(0)}\) from the dataset
      \For{\(k = 2\) to \(K\)}
      \State Compute \( D_i = \min_{1 \le j < k} d(x_i, \mu_j^{(0)})^2 \) for each \( i = 1, \ldots, n \)
      \State Sample the centroid \(\mu_k^{(0)}\) from \(\{x_i\}\) with probability proportional to \(D_i\)
      \EndFor

      \For{\( t = 1 \) to \( T \)}
      \State \textbf{Assignment step:}
      \For{each data point \( x_i \)}
      \State Assign \( x_i \) to cluster \( C_k^{(t)} \), with \( k = \arg\min_{j=1,\ldots,K} d(x_i, \mu_j^{(t-1)}) \)
      \EndFor
      \State \textbf{Update step:}
      \For{each cluster \( k = 1, \ldots, K \)}
      \State \( \bar{z}_k^{(t)} \gets \dfrac{1}{|C_k^{(t)}|} \sum_{x_i \in C_k^{(t)}} z_i \) \Comment{eq.~\eqref{eq:centroid_linear}}
      \State \( \bar{\theta}_k^{(t)} \gets \arg\min_{\vartheta \in [0, 2\pi)} \sum_{x_i \in C_k^{(t)}} \delta(\theta_i, \vartheta)^2 \) \Comment{see below}
      \State Set centroid \( \mu_k^{(t)} = (\bar{\theta}_k^{(t)}, \bar{z}_k^{(t)}) \)
      \EndFor
      \If{cluster assignments \( \{C_k^{(t)}\} \) have converged or \( t = T \)}
      \State \textbf{break}
      \EndIf
      \EndFor
      \State \Return final cluster assignments \( \{C_k^{(t)}\} \) and centroids \( \{\mu_k^{(t)}\} \)
      \Statex \hrulefill
      \begin{itemize}
        \item if \( d = d_{\mathrm{geo}} \), then \( \delta = \delta^{\mathrm{geo}} \) and \( \bar{\theta}_k^{(t)} \) is the Fr\'echet mean of the arc, eq.~\eqref{eq:centroid_geo}, computed by Algorithm~\ref{alg:frechet_circular};
        \item if \( d = d_{\mathrm{chord}} \), then \( \delta = \delta^{\mathrm{chord}} \) and \( \bar{\theta}_k^{(t)} \) is the mean direction, eq.~\eqref{eq:centroid_meandir};
        \item the radius \( r \) enters only through \( d \), in the initialization and in the assignment step: the update step does not depend on it.
      \end{itemize}
    \end{algorithmic}
  \end{spacing}
\end{algorithm}

\medskip
\begin{remark}\label{remark:arc_vs_chord}
  The two variants differ only in how far a point is deemed to be from its centroid along the circle. For an angular separation \(\Delta = \delta^{\mathrm{geo}}(\theta_i, \bar\theta) \in [0,\pi]\), and up to the common factor \(r^2\), the arc charges \(\Delta^2\) and the chord \(\big(\delta^{\mathrm{chord}}\big)^2 = 4\sin^2(\Delta/2) = \Delta^2 - \Delta^4/12 + O(\Delta^6)\): the two agree to second order and coincide, for practical purposes, in the concentrated regime, but the chordal cost saturates at \(4\) for \(\Delta = \pi\), whereas \(\big(\delta^{\mathrm{geo}}\big)^2\) attains \(\pi^2 \approx 9.87\). The chord therefore systematically under-penalizes points far from the centroid, and the two angular centroids move apart accordingly; the resulting difference between the mean direction and the Fr\'echet mean of the arc is the subject of \citet{McKilliam2012} and \citet{HotzHuckemann2015}, and Appendix~\ref{sec:arc_chord} illustrates it on a worked example. The radius \(r\), in turn, sets the relative weight of the two coordinates under either distance: its choice is application-specific and is not addressed in this paper, where we set \(r = 1\) throughout.
\end{remark}

\paragraph*{Computational cost.}
All the ingredients of Algorithm~\ref{alg:geodesic_kmeans} have known costs. The \(K\)-means++ seeding costs \(O(nK)\) when \(D_i\) is carried along as a running minimum \citep{arthur2006k}, and is paid once. Each iteration evaluates \(nK\) distances, each in \(O(1)\) operations under either metric, and then updates the centroids: the arithmetic mean and the mean direction cost \(O(|C_k|)\) per cluster, the Fr\'echet mean of the arc \(O(|C_k| \log |C_k|)\), the cost of sorting the angles \citep{HotzHuckemann2015, Cazals2021}. An iteration therefore costs \(O(nK)\) for the chordal \(K\)-means and \(O(nK + n \log n)\) for the geodesic one, that is, the cost of the classical \(K\)-means as soon as \(K \gtrsim \log n\).

\subsection{\texorpdfstring{Convergence of the \(K\)-means algorithm}{Convergence of the K-means algorithm}}\label{subsec:convergence}

The convergence argument of the classical \(K\)-means carries over to Algorithm~\ref{alg:geodesic_kmeans} unchanged, because it uses only two facts: the assignment step minimizes the objective for fixed centroids, and the update step returns an exact Fr\'echet mean of the distance in use. Both hold for \(d_{\mathrm{geo}}\) and for \(d_{\mathrm{chord}}\).

\begin{proposition}\label{prop:convergence}
  Let \(\{x_i\}_{i=1}^n\) be a finite dataset on \(\mathcal{C}\) and let \(d \in \{d_{\mathrm{geo}}, d_{\mathrm{chord}}\}\), for any fixed radius \(r > 0\). Consider Algorithm~\ref{alg:geodesic_kmeans}, which alternates between assigning points to the nearest centroid under \(d\) and updating centroids as Fr\'echet means of the assigned points. Then the sequence of objective values
  \[
  J^{(t)} = \sum_{k=1}^K \sum_{x_i \in C_k^{(t)}} d(x_i, \mu_k^{(t)})^2
  \]
  generated by the algorithm satisfies \(J^{(t+1)} \leq J^{(t)}\) for all \(t\), and, provided points are reassigned only when this strictly decreases \(J\) and ties in the centroid update are resolved by a fixed rule, the algorithm converges in a finite number of iterations to a local minimum or stationary point of \(J\).
\end{proposition}

\noindent\textbf{Proof.}

\begin{enumerate}
  \item \emph{Monotonicity:} At iteration \(t\), given the centroids \(\{\mu_k^{(t-1)}\}_{k=1}^K\), the assignment step defines clusters
        \[
        C_k^{(t)} = \{x_i : k = \arg\min_j d(x_i, \mu_j^{(t-1)})\}.
        \]
        By construction, this assignment minimizes the within-cluster sum of squared distances for fixed centroids, hence
        \[
        \sum_{k=1}^K \sum_{x_i \in C_k^{(t)}} d(x_i, \mu_k^{(t-1)})^2 \leq \sum_{k=1}^K \sum_{x_i \in C_k^{(t-1)}} d(x_i, \mu_k^{(t-1)})^2 = J^{(t-1)}.
        \]

  \item \emph{Centroid update:} For fixed clusters \(\{C_k^{(t)}\}\), the centroid update sets
        \[
        \mu_k^{(t)} = \arg\min_{\mu \in \mathcal{C}} \sum_{x_i \in C_k^{(t)}} d(x_i, \mu)^2,
        \]
        the Fr\'echet mean of the cluster. By the additive separability of \(d^2\), this minimizer is attained exactly at the pair given by~\eqref{eq:centroid_linear} and by~\eqref{eq:centroid_geo} or~\eqref{eq:centroid_meandir}: the Fr\'echet mean of the arc is computed exactly by Algorithm~\ref{alg:frechet_circular}, as established in Proposition~\ref{prop:frechet}, and the mean direction is the exact minimizer of the chordal cost whenever \(R_k > 0\), every direction being a minimizer when \(R_k = 0\). Hence
        \[
        \sum_{x_i \in C_k^{(t)}} d(x_i, \mu_k^{(t)})^2 \leq \sum_{x_i \in C_k^{(t)}} d(x_i, \mu_k^{(t-1)})^2 .
        \]
        Summing over all \(k\):
        \[
        J^{(t)} = \sum_{k=1}^K \sum_{x_i \in C_k^{(t)}} d(x_i, \mu_k^{(t)})^2 \leq \sum_{k=1}^K \sum_{x_i \in C_k^{(t)}} d(x_i, \mu_k^{(t-1)})^2.
        \]

  \item \emph{Combining steps:} Combining the two inequalities yields \(J^{(t)} \leq J^{(t-1)}\). Thus \(\{J^{(t)}\}\) is a monotonically non-increasing sequence bounded below by zero.

  \item \emph{Finite convergence:} There are finitely many clusterings of a finite dataset, and by hypothesis every change of clustering strictly decreases \(J\); no clustering can therefore recur, and the algorithm terminates after finitely many iterations at some iteration \(T\) where
        \[
        \{C_k^{(T)}\} = \{C_k^{(T-1)}\}, \quad \{\mu_k^{(T)}\} = \{\mu_k^{(T-1)}\}.
        \]
        At this point, \(J^{(T)}\) is a local minimum or stationary value of the objective.
\end{enumerate}

\qed

Note that if the cylinder is replaced by \(\mathbb{R}^d\) with the Euclidean metric, Algorithm~\ref{alg:geodesic_kmeans} reduces to the classical \(K\)-means algorithm, with cluster centroids as arithmetic means and objective function as the sum of squared Euclidean distances.

\section{Simulations}\label{sec:simulations}
To investigate the performances of the proposed clustering method we perform an extensive analysis on simulated data, considering varying number of clusters at different level of overlap and geometrical structure.
We cluster the data using our proposed method as well as alternative ones from the literature, and compare the ability of the different methodologies to retrieve the true clustering structure, reporting performance in terms of adjusted Rand Index.
Subsection~\ref{subsec:dgp_description} describes the data generating processes, Subsection~\ref{subsec:sim_methods} the compared clustering methodologies and Subsection~\ref{subsec:sim_results} discusses the results.

\subsection{Data generating process}\label{subsec:dgp_description}

Simulated data are generated from a mixture distribution with Mardia--Sutton components~\citep{MardiaEtAl1979}.
For each design we fix $K$ and the cluster sizes $n_k$, $k = 1, \ldots, K$, and the $n_k$ points of cluster $k$ are drawn as follows:
(i) the angular coordinate, $\theta$, is generated from a von Mises distribution with mean $\mu_{0,k}$ and concentration $\kappa_{k}$; (ii) the linear coordinate is generated conditional on $\theta$ from a normal distribution with parameters
\begin{equation*}
\begin{matrix}  \mu_{c,k} = \mu_k + \sqrt{\sigma_k^2\kappa_k}\left(\rho_{1,k}(\cos(\theta)-\cos(\mu_{0,k})) + \rho_{2,k}(\sin(\theta)-\sin(\mu_{0,k}))\right),\\
    \sigma_{c,k}^2 =  \sigma^2_k(1-(\rho_{1,k}^2 + \rho_{2,k}^2)),
\end{matrix}
\end{equation*}
where $\mu_k$ and $\sigma^2_k$ are the unconditional mean and variance of the linear coordinate in cluster $k$, and $\rho_{1,k}$ and $\rho_{2,k}$ govern the correlation between the angular and the linear component, with $\sqrt{\rho_{1,k}^2 + \rho_{2,k}^2} < 1$.

We consider ten designs, all with balanced clusters of $n_k = 100$ points: eight are obtained by crossing the factors reported in panel~(a) of Table~\ref{tab:simconfigs}, and two more are described in panel~(b) of the same table.

\begin{table}[!htpb]
  \centering
  \caption{Configuration of the simulated designs. Panel~(a) reports the factors
    that are crossed to obtain the eight factorial designs, labeled
    \texttt{[A][L]K}: the first letter is the angular dispersion, the second the
    linear variance, and $K\in\{3,6\}$ the number of clusters. Panel~(b) reports the
    two designs that fall outside the factorial, in which cluster shape is driven by
    a strong cluster-specific angular--linear correlation. Parameters are those of
    the Mardia--Sutton components described in Section~\ref{subsec:dgp_description};
    $\mu_0$ is the angular mean, $\mu$ the unconditional linear mean, $\kappa$ the
    angular concentration and $\sigma^2$ the unconditional linear variance. All
    designs feature balanced clusters with $n_k=100$ observations each.}
  \label{tab:simconfigs}
  \small
  \begin{tabular}{@{}llp{0.60\textwidth}@{}}
    \toprule
    \multicolumn{3}{@{}l}{\textit{(a) Factorial designs}: all $2\times2\times2$ combinations of the factors below}\\
    \midrule
    Cluster centers & $K=3$ & $(\mu_0,\mu)=(\tfrac{7}{8}\pi,0),\ (\tfrac{5}{8}\pi,2),\ (\tfrac{3}{8}\pi,4)$\newline
                              angular means $157.5^{\circ},112.5^{\circ},67.5^{\circ}$: a $90^{\circ}$ arc away from the junction\\
    \addlinespace
                    & $K=6$ & $(\mu_0,\mu)=(\tfrac{13}{12}\pi,0),\ (\tfrac{9}{12}\pi,2),\ (\tfrac{5}{12}\pi,4),$\newline
                              $\phantom{(\mu_0,\mu)=}(\tfrac{1}{12}\pi,4),\ (\tfrac{21}{12}\pi,2),\ (\tfrac{17}{12}\pi,0)$\newline
                              angular means equispaced $60^{\circ}$ apart from $15^{\circ}$; each linear level shared by two components\\
    \addlinespace
    Angular dispersion & \texttt{L} & $\kappa=10$ \quad (concentrated)\\
                       & \texttt{H} & $\kappa=0.5$ \quad (nearly uniform on $S^1$)\\
    \addlinespace
    Linear variance & \texttt{L} & $\sigma^{2}=0.1$\\
                    & \texttt{H} & $\sigma^{2}=2$\\
    \addlinespace
    Angular--linear correlation & --- & $\rho_1=\rho_2=0$\\
    \midrule
    \multicolumn{3}{@{}l}{\textit{(b) Designs outside the factorial}}\\
    \midrule
    \texttt{COR3} & $K=3$ & $\mu_0=(0,\tfrac{1}{3}\pi,\tfrac{11}{6}\pi)$,\quad $\mu=(0,0.5,0.5)$,\quad $\kappa=10$,\quad $\sigma^{2}=0.5$\newline
                            $(\rho_1,\rho_2)=(0.60,0.75),\,(-0.68,-0.68),\,(-0.75,-0.60)$\\
    \addlinespace
    \texttt{COR6} & $K=6$ & $\mu_0$ equispaced $60^{\circ}$ apart from $0$,\quad $\mu=(0,1,0,1,0,1)$,\quad $\kappa=8$,\quad $\sigma^{2}=0.6$\newline
                            $(\rho_1,\rho_2)=(0.60,0.75),\,(0.68,0.68),\,(0.75,0.60)$, cycled over the six components\\
    \bottomrule
  \end{tabular}
\end{table}

The eight factorial designs have uncorrelated coordinates, \ie $\rho_{1,k} = \rho_{2,k} = 0$ for all $k$,\footnote{In unreported results, we also tested moderate and high correlation values, with $\rho_{1,k} = \rho_{2,k}$ and $\sqrt{\rho_{1,k}^2 + \rho_{2,k}^2} \in \{0.5, 0.9\}$, and found results qualitatively similar to those of the uncorrelated designs. Thus, only the latter are reported in the main text to ease the discussion.} and are labeled by a three-character code: the first letter is the angular dispersion (\texttt{L}: $\kappa = 10$, concentrated; \texttt{H}: $\kappa = 0.5$, nearly uniform on $S^1$), the second the linear variance (\texttt{L}: $\sigma^2 = 0.1$; \texttt{H}: $\sigma^2 = 2$), and the last character is $K$.
The cluster centers are fixed for each value of $K$ and are chosen so that the role of the circular coordinate differs sharply between the two cases.
With $K=3$, the groups are separated along both coordinates and occupy a $90^{\circ}$ arc that stays away from the $0/2\pi$ junction, so that, at low angular dispersion, the periodicity of $\theta$ is inactive by construction.
With $K=6$, the angular means are equispaced $60^{\circ}$ apart starting at $15^{\circ}$, so that some clusters straddle the junction, and each linear mean is shared by two components, so that the linear coordinate alone cannot complete the partition and the angle must be used.
The two designs of panel~(b), \texttt{COR3} and \texttt{COR6}, place the clusters at intermediate dispersion and endow them with a strong, cluster-specific angular--linear correlation, $\sqrt{\rho_{1,k}^{2}+\rho_{2,k}^{2}}\approx0.96$, with partially overlapping linear means: the groups are markedly elongated along distinct, cluster-dependent directions and genuinely interleave on the cylinder, and their shape is driven by the correlation rather than by the dispersion regime.

Figure~\ref{fig:simdata} shows one dataset per design and makes the intended contrasts visible.
The \texttt{LL} designs are the well-separated ones: \texttt{LL3} is trivial, whereas \texttt{LL6} is the only factorial design in which the circular topology is both active and recoverable, so that the treatment of the periodicity is the sole difference left between competing criteria.
In \texttt{HL3} the angle is nearly pure noise while the linear coordinate separates the groups on its own, so the design measures how much weight each criterion concedes to an uninformative angle; in \texttt{HL6} the shared linear levels remove that separation and the dispersed angle becomes indispensable.
The \texttt{LH} designs are the mirror image, an informative angle blurred by a noisy linear coordinate, and the \texttt{HH} designs are noise-dominated in both coordinates.
\texttt{COR3} and \texttt{COR6}, finally, are the only designs whose cluster shape violates the isotropy implicit in the squared-distance criteria.
From each design we generate $100$ independent datasets, and Section~\ref{subsec:sim_results} aggregates the results over the replicates.

\begin{figure}[!htpb]
  \centering
  \includegraphics[width=0.9\textwidth]{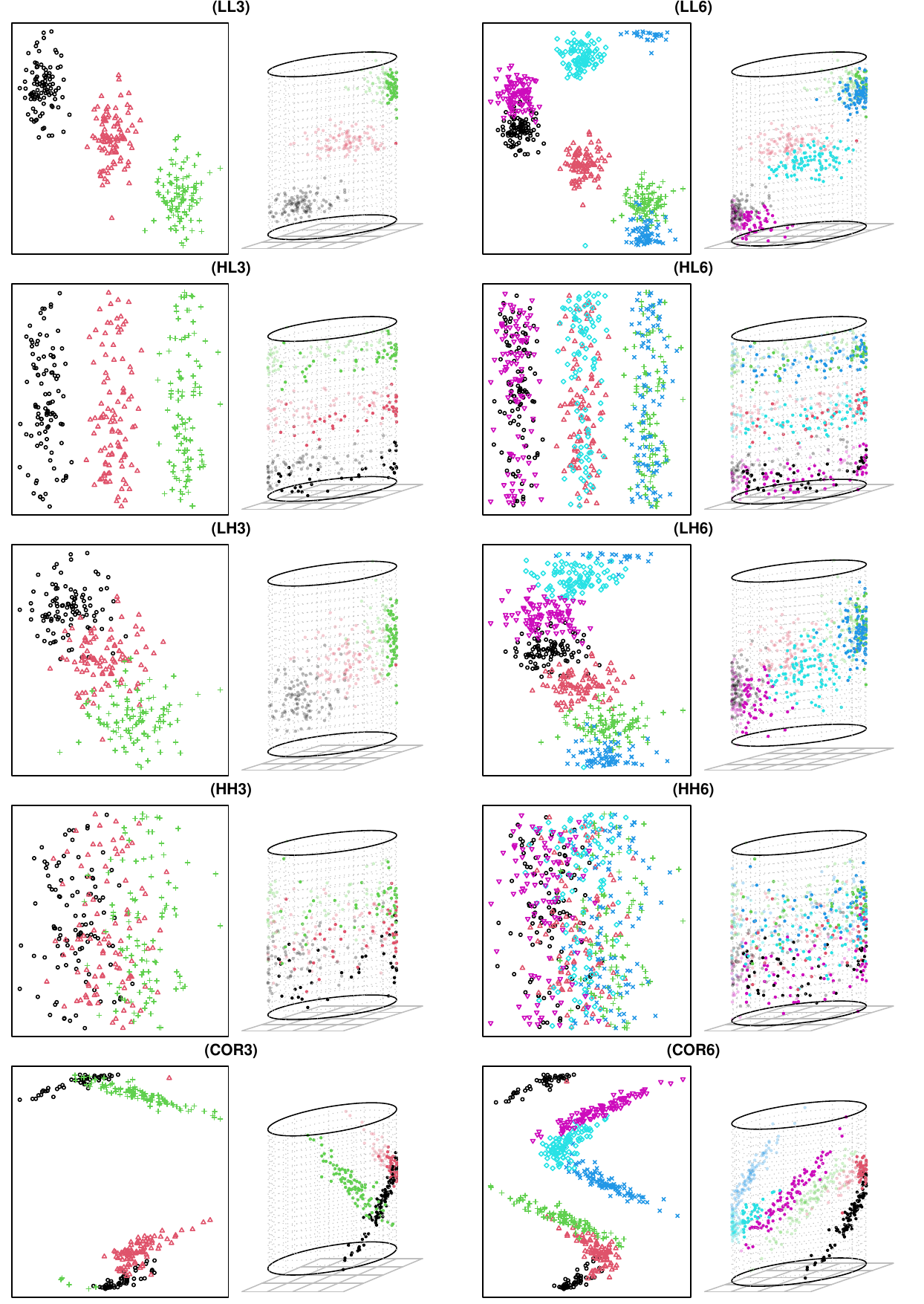}
  \caption{One simulated dataset per design. Each pair of panels shows the same
    dataset as a planar scatter of the linear coordinate against the angular one
    (left) and on the lateral surface of the cylinder (right), with points on the far
    side of the cylinder shaded. Colors and symbols mark the true
    clusters.\label{fig:simdata}}
\end{figure}

\subsection{Methods under comparison}\label{subsec:sim_methods}
We consider four clustering methodologies: \euckmeans, \chordkmeans, \geokmeans and the model-based mixture model introduced in~\cite{sengupta2020model}, \sengupta.

\begin{itemize}
  \item The \euckmeans is the standard Lloyd k-means algorithm~\citep{Lloyd1957}, where the angular and linear components are treated as Euclidean coordinates, the periodicity of \(\theta\) being ignored altogether.
  \item The \chordkmeans is Algorithm~\ref{alg:geodesic_kmeans} with the chordal distance \(d_{\mathrm{chord}}\) of~\eqref{eq:dchord}, the angular centroid being the mean direction~\eqref{eq:centroid_meandir}.
  \item The \geokmeans is Algorithm~\ref{alg:geodesic_kmeans} with the geodesic distance \(d_{\mathrm{geo}}\) of~\eqref{eq:dgeo}, the angular centroid being the Fr\'echet mean of the arc~\eqref{eq:centroid_geo}, computed by Algorithm~\ref{alg:frechet_circular}. It differs from \chordkmeans in the metric alone.
  \item Finally, \sengupta implements Model~1 from~\cite{sengupta2020model}, and it is a model-based mixture using Mardia-Sutton density to model mixture components. The model's parameters are estimated via an Expectation Maximization algorithm, initialized at random.
\end{itemize}

This panel is kept fixed throughout the paper: the applications of Section~\ref{sec:images} and Section~\ref{sec:case1} compare the same four methods, the former adding an application-specific baseline.

The k-means based algorithms are initialized with the popular ``k-means++'' initialization, introduced in~\cite{arthur2006k}, each under its own distance.
Each method is restarted independently five times with at most \(300\) iterations, and the best solution in terms of the corresponding cost function (or the log-likelihood, for the model-based method) is retained.%
%\footnote{The iteration cap is immaterial for the three \(K\)-means variants, whose objective is unchanged at \(100\) and at \(300\) iterations, and binds only on the EM recursion of \sengupta: in the concentrated designs (\(\kappa = 10\)) the latter is still at the cap at every \(K\), and remains so at \(1000\) iterations for all but one value of \(K\). The fits of \sengupta at large \(K\) are therefore, if anything, slightly pessimistic.}

For each data set, we fit the clustering method fixing the true number of clusters.
The solutions are assessed measuring their adjusted Rand Index~(ARI; \cite{Rand1971JotASA}) against the true partition.

\subsection{Discussion of the results}\label{subsec:sim_results}
Table~\ref{tab:ari} reports the ARI averaged over the $100$ Monte Carlo replications of each design, with its standard deviation, and Figure~\ref{fig:ari_boxplot} the corresponding box plots, which make the variability of the solutions visible.%
%\footnote{The eight factorial designs were also generated at two positive correlation levels, $\rho = 0.5$ and $\rho = 0.9$ with $\rho_1 = \rho_2 = \rho/\sqrt{2}$: the ranking of the methods is unchanged in every cell up to differences within Monte Carlo error, the only movements being the rise of \sengupta in \texttt{HH3} (from $0.118$ to $0.252$, with a standard deviation five times that of the $K$-means variants) and the erosion, never a reversal, of the advantage of \chordkmeans in \texttt{HL3} ($0.805$ to $0.598$); those replicates are therefore not reported.}
We discuss first the three designs singled out in Section~\ref{subsec:dgp_description}, then the remaining ones.

\begin{figure}[p]
	\centering
  \includegraphics[width=0.9\textwidth]{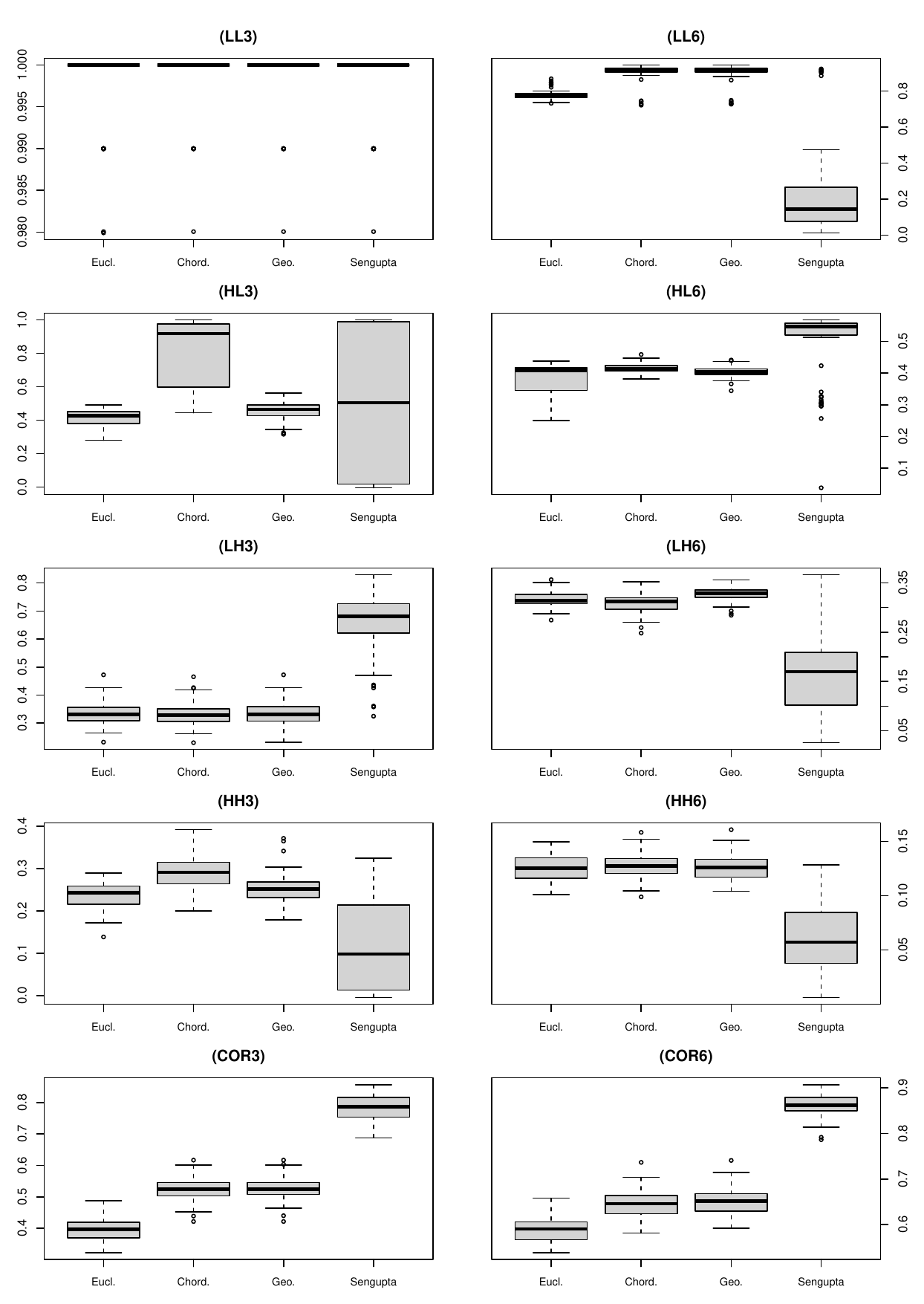}
	\caption{Distribution of the adjusted Rand Index over the $100$ Monte Carlo
    replications of each design, by method. Panels follow the order of
    Table~\ref{tab:ari}: designs with $K=3$ on the left, $K=6$ on the
    right.\label{fig:ari_boxplot}}
\end{figure}

\begin{table}
  \centering
  \caption{Average ARI over 100 Monte Carlo replications (standard deviation in
    parentheses), at zero angular--linear correlation. Design codes as in
    Table~\ref{tab:simconfigs}; the two \texttt{COR} rows are the correlated designs of
    its panel~(b). The best result in each row is in bold.}\label{tab:ari}
  \begin{tabular}{@{}lcccc@{}}
    \toprule
    Design & Eucl. K-means++ & Chord. K-means++ & Geo. K-means++ & Sengupta et al.\\
    \midrule
    \texttt{LL3}   & 0.998 (0.004) & \textbf{0.999} (0.003) & \textbf{0.999} (0.003) & 0.999 (0.004)\\
    \texttt{LL6}   & 0.779 (0.026) & \textbf{0.909} (0.043) & 0.908 (0.042) & 0.213 (0.219)\\
    \midrule
    \texttt{HL3}   & 0.415 (0.048) & \textbf{0.805} (0.204) & 0.456 (0.057) & 0.468 (0.393)\\
    \texttt{HL6}   & 0.384 (0.047) & 0.415 (0.013) & 0.403 (0.016) & \textbf{0.506} (0.097)\\
    \midrule
    \texttt{LH3}   & 0.335 (0.042) & 0.330 (0.042) & 0.335 (0.042) & \textbf{0.659} (0.099)\\
    \texttt{LH6}   & 0.317 (0.014) & 0.309 (0.018) & \textbf{0.328} (0.014) & 0.163 (0.075)\\
    \midrule
    \texttt{HH3}   & 0.238 (0.028) & \textbf{0.290} (0.036) & 0.251 (0.031) & 0.118 (0.103)\\
    \texttt{HH6}   & 0.125 (0.011) & \textbf{0.128} (0.011) & 0.126 (0.011) & 0.061 (0.029)\\
    \midrule
    \texttt{COR3}  & 0.397 (0.034) & 0.526 (0.035) & 0.527 (0.034) & \textbf{0.783} (0.041)\\
    \texttt{COR6}  & 0.589 (0.027) & 0.645 (0.030) & 0.651 (0.030) & \textbf{0.861} (0.024)\\
    \bottomrule
  \end{tabular}
\end{table}

On \texttt{LL6}, the two cylindrical variants recover the partition almost completely ($0.909$ and $0.908$), while \euckmeans stops at $0.779$.
Since the three procedures differ only in the treatment of the angular component, the gap is the cost of ignoring the periodicity, and it appears as soon as clusters straddle the $0/2\pi$ junction.
That \chordkmeans and \geokmeans are indistinguishable is equally informative: at $\kappa = 10$ the within-cluster angular differences are small, arc and chord agree to second order (Remark~\ref{remark:arc_vs_chord}), and the choice of the metric is immaterial.
\sengupta collapses to $0.213$ despite the favorable separation, with the EM recursion frequently trapped in poor local optima.

In the \texttt{HL3} design, the angle bears no information, and the linear means $0, 2, 4$ identify the groups on their own.
\euckmeans drops to $0.415$: near the junction the unwrapped angular differences reach $(2\pi)^2$ and corrupt the assignment.
The wrapping of \geokmeans removes that artefact but not the noise: the squared arc still ranges over $[0, \pi^2]$, up to two and a half times the squared gap of $4$ between adjacent linear centers, and \geokmeans improves only to $0.456$.
The squared chordal distance, used by \chordkmeans, saturates at $4$ (exactly the squared linear gap) and the angular noise no longer override the linear separation. \chordkmeans achieves the best ARI at $0.805$ (its boxplot is nonetheless wide).
Here, the under-penalization of distant points of the chordal distance turns into an accidental form of robustness to the uninformative angular coordinate.
\sengupta, whose down-weighting is explicit and adaptive, is unreliable here: the EM search either finds the optimum or misses it altogether, with replications either close to $1$ or close to $0$ (mean $0.468$, standard deviation $0.393$).

For the elongated and correlated clusters \texttt{COR3} and \texttt{COR6}, no squared-distance criterion can follow a cluster-specific direction, and the three $K$-means variants stay between $0.40$ and $0.65$, the two cylindrical ones again indistinguishable ($0.527$ and $0.526$ at $K = 3$, $0.651$ and $0.645$ at $K = 6$) and clearly ahead of \euckmeans ($0.397$ and $0.589$).
\sengupta, which fits a separate conditional linear mean $\mu_{c,k}$ to each component, wins decisively ($0.783$ and $0.861$) and stably.

The other six designs tell a common story. 
The three $K$-means variants yield similar results, suggesting that the designs bear no mechanism able to distinguish them, unlike the previous settings; \sengupta is either the best or the worst, and almost always the most dispersed.
In \texttt{LL3} every method is essentially perfect ($0.998$--$0.999$).
In \texttt{HL6}, the shared linear levels force every method to rely on a dispersed angle: the three $K$-means achieve a poor recovery ($0.384$--$0.415$) and \sengupta is ahead ($0.506$) but with a long lower tail.
In \texttt{LH3} the noisy linear coordinate dominates the squared distance and the three variants collapse together at $0.33$, the concentrated angle keeping the metrics in their small-angle regime; \sengupta, which estimates a separate dispersion for each coordinate, down-weights the linear one and recovers the structure ($0.659$), at the price of a wider spread and outliers down to $0.3$.
At $K = 6$ (\texttt{LH6}) the same likelihood becomes multimodal, \sengupta falls to $0.163$ with replications from $0$ to $0.35$, and the $K$-means variants regain a modest lead, \geokmeans marginally ahead ($0.328$ against $0.317$ and $0.309$).
The \texttt{HH} designs leave little to recover, and every method approaches chance, while \sengupta is the weakest and the most variable.

Overall, the geodesic metric is the exact geometry of the cylinder and remains the default choice; the chordal one overtakes it only when the angle carries no information, its cap acting as an implicit, fixed down-weighting of the angular coordinate.
That advantage is a matter of scale rather than of geometry.
\sengupta, whose failures are ones of the EM search rather than of the model is erratic in most of them, but on the elongated, cluster-specific shapes that violate the isotropy of any squared-distance criterion (\texttt{COR}) it wins decisively and stably.

\section{Application to color image quantization}\label{sec:images}

Color quantization reduces an image to a palette of $K$ colors: the pixels are
grouped, and every pixel is repainted with the representative color of its group.
It is a clustering problem in color space, and a natural test bed for cylindrical
methods, because the hue--saturation--value (HSV) representation of color
\citep{smith1978color} places the hue on a circle.

\subsection{Data}\label{subsec:images_data}

The images come from the Berkeley Segmentation Data Set BSDS300
\citep{martin2001database}, a standard collection of $300$ natural color photographs
of $481\times321$ pixels ($n = 154{,}401$ points per image).
The same data set was used to evaluate the cylindrical clustering method proposed by \cite{fujita2017clustering}.
We select six images so as to span different chromatic structures.
Write $\theta_j = 2\pi H_j$ for the hue of pixel $j$ read as an angle and $x_j = V_j$ for its value, so that a pixel is the point $(\theta_j, x_j)$ on the cylinder.
Table~\ref{tab:images} characterizes the six images through four summaries of that representation, which are informative to compare the clustering methods.

The circular variance of the hue, $1-\bar R$ with $\bar R$ the mean resultant length
of Section~\ref{subsec:method}, says whether the hue carries in any information at all.\footnote{$\bar{R}$ is computed as the resultant length $R_k$ defined in Section~\ref{subsec:method} divided by the number of pixels, so that it falls in $[0,1]$. $\bar R = 1$ when every pixel shares one hue, and $\bar R = 0$ when the hues are spread evenly round the circle.}
The ``wrap ratio'', the linear variance of the hue over twice its circular variance, is
inflated whenever a compact group of hues straddles the $0/2\pi$ cut, and so says in
advance on which images treating the hue as periodic can matter.\footnote{The wrap
  ratio is not a standard summary of circular data; we use it only as a quick way of
  locating where the periodicity is active. Expanding about the mean direction gives
  $1-\bar R \simeq \operatorname{var}(\theta)/2$ for concentrated hues, so that
  $\operatorname{var}(\theta)/\{2(1-\bar R)\}$ sits near $1$ whenever no group of hues
  crosses the cut, and departs from it when one does. Five of the six images of
  Table~\ref{tab:images} lie between $1.03$ and $1.21$; \texttt{124084} is at
  $10.37$.} The $R^2$
of the regression of the value on $(\cos\theta,\sin\theta)$ measures the hue--value
dependence that \sengupta is built to exploit. 
The fraction of nearly achromatic pixels ($\mathrm{Saturation}<0.15$) locates where the hue is ill-defined, the saturation being the distance of a color from the achromatic axis, and the cylindrical representation noisy by construction.

The set thus contains images whose hues are widely dispersed
(\texttt{118035}, \texttt{299091}) and others concentrated in a narrow arc
(\texttt{25098}, \texttt{12003}, \texttt{113044}), one image whose dominant hue sits
on the cut (\texttt{124084}, wrap ratio $10.37$ against at most $1.21$ for every other
image), and images spanning the whole range of hue--value coupling, from
\texttt{118035} ($R^2 = 0.78$) to \texttt{113044} ($R^2 = 0.03$).

\begin{table}[!htpb]
  \centering
  \caption{The six BSDS300 images and their chromatic summaries in the hue--value
    representation: circular variance of the hue, wrap ratio
    $\operatorname{var}(\theta)/\{2(1-\bar R)\}$, $R^2$ of the value on
    $(\cos\theta,\sin\theta)$, and fraction of pixels with saturation below
    $0.15$.}\label{tab:images}
  \footnotesize
  \begin{tabular}{@{}llcccc@{}}
    \toprule
    Image & Subject & Circ.\ var. & Wrap ratio & $R^2$ & Low-$S$ frac.\\
    \midrule
    \texttt{118035} & church, red dome against a blue sky       & 0.80 &  1.21 & 0.78 & 0.18\\
    \texttt{25098}  & market stall: red peppers, wooden crates  & 0.06 &  1.03 & 0.26 & 0.02\\
    \texttt{124084} & red flowers with yellow centers on leaves & 0.17 & 10.37 & 0.32 & 0.00\\
    \texttt{299091} & pyramid and desert under a clouded sky    & 0.35 &  1.13 & 0.40 & 0.00\\
    \texttt{113044} & chestnut mare and foal in a green field   & 0.15 &  1.09 & 0.03 & 0.00\\
    \texttt{12003}  & orange sea star on green algae            & 0.12 &  1.06 & 0.43 & 0.01\\
    \bottomrule
  \end{tabular}
\end{table}

\subsection{Methodology}\label{subsec:images_methods}

Each pixel is an RGB triplet with channels in $[0,1]$, which we convert to HSV. The
hue $H\in[0,1)$ is an angle, and the value $V\in[0,1]$ a linear coordinate, so that
the pair
\begin{equation*}
  (\theta, x) = (2\pi H,\ V) \in S^1\times[0,1]
\end{equation*}
is a point on the lateral surface of the cylinder. The saturation, which only
measures how far a color lies from the achromatic axis, takes no part in the
cylindrical clusterings and is reinstated at painting time: every pixel of cluster $k$
is repainted with the color $(H_k, S_k, V_k)$, where $H_k$ and $V_k$ are read off the
cluster centroid $(\bar\theta_k, \bar x_k)$ and $S_k$ is the mean saturation of the
pixels in the cluster. 
Every method therefore paints exactly $K$ whole colors.
For \sengupta, whose components have no centroid, we take the angular mean $\mu_k$ of the component and its conditional linear mean at that angle.
%, $b_{0k}+b_{1k}\cos\mu_k+b_{2k}\sin\mu_k$.

We compare the four methods of Section~\ref{subsec:sim_methods} -- \euckmeans,
\chordkmeans, \geokmeans and \sengupta, all run on $(\theta,x)$ -- and a fifth one,
\rgbkmeans, which is the standard $k$-means++ applied to the raw RGB triplets, which ignores the circular structure of the hue but retains all three channels.
The search budget is the same for every method: $k$-means++ seeding with
$50$ restarts and at most $300$ iterations for the four $K$-means variants, $50$
random initializations and $300$ EM iterations for \sengupta, the best solution in
terms of the method's own criterion being retained. We fix $K=3$ for all methods.

\subsection{Results}\label{subsec:images_results}

Figure~\ref{fig:images} shows, for each image, the original and the five quantized
versions. Three comparisons can be read off it.

\begin{figure}[p]
  \centering
  \includegraphics[width=0.96\textwidth]{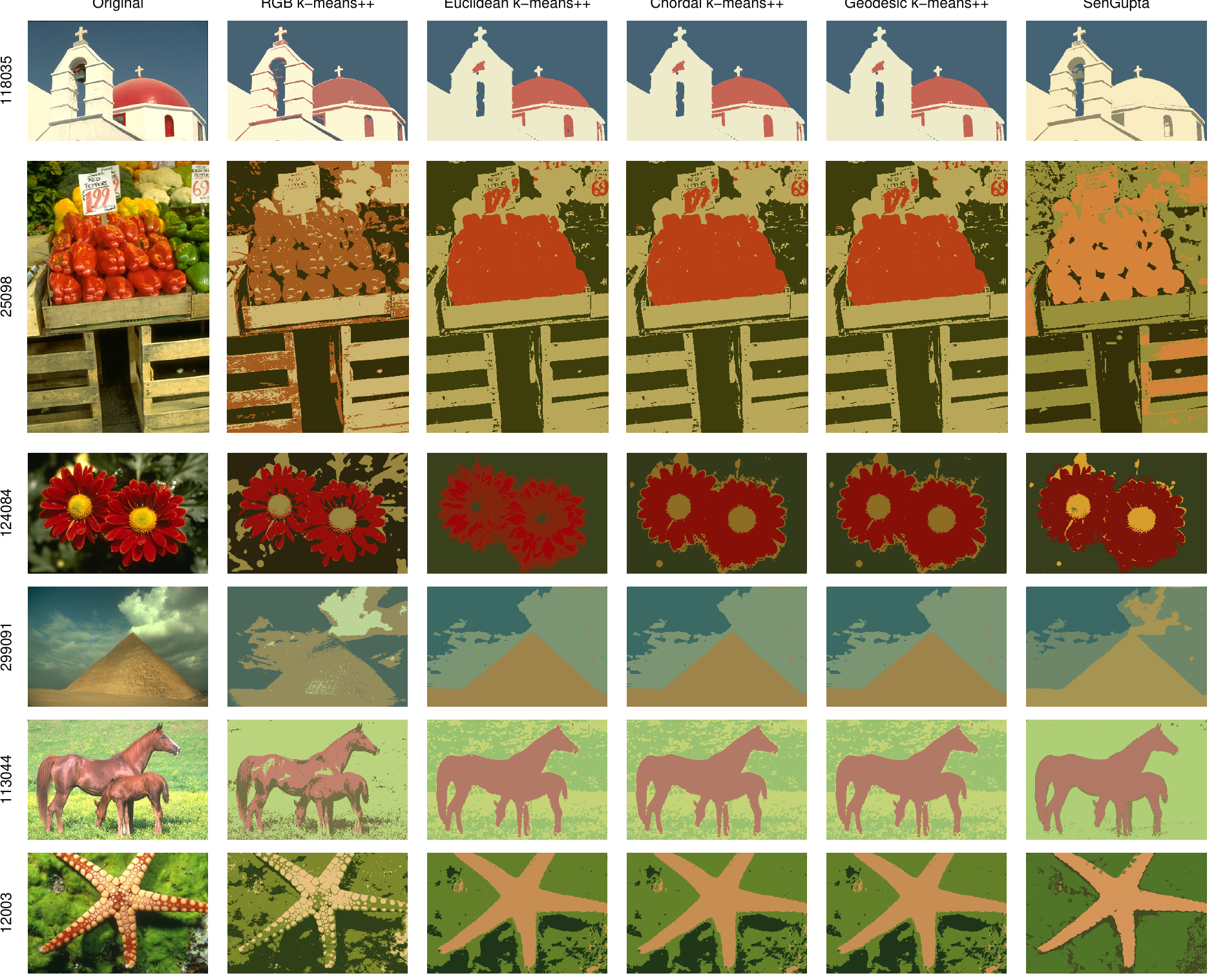}
  \caption{Color quantization at $K=3$. The first column shows the original BSDS300
    image; the remaining columns the image repainted with the three colors found by
    each method. The four cylindrical methods work on (hue, value); \rgbkmeans on the
    raw RGB triplets.\label{fig:images}}
\end{figure}

The first is between \euckmeans and the two cylindrical $K$-means, which differ only
in whether the hue is treated as periodic. 
On the five images whose hues stay away from the $0/2\pi$ cut (\texttt{118035}, \texttt{25098}, \texttt{299091}, \texttt{113044}, \texttt{12003}) the three quantizations are indistinguishable, as expected.
In \texttt{124084}, whose dominant hue sits on the cut (wrap ratio $10.37$ against at most $1.21$ for the others), \euckmeans splits in two reds, so that the yellow centers are lost and confounded with the background, in contrast with \chordkmeans and \geokmeans returning the natural palette: petals, centers, leaves. 
This is the image-space counterpart of the \texttt{LL6} design of Section~\ref{sec:simulations}.

The second comparison is between \chordkmeans and \geokmeans, which are
indistinguishable on all six images. The within-cluster hue dispersion of a
quantized image is small by construction, so the chord and the arc agree to second
order and the two metrics yield the same partition, as in the concentrated
designs of the simulation study.

The third comparison is with the two methods that do not work on the cylinder in the
same way. \rgbkmeans spends its three colors on luminance rather than on hue
(\texttt{25098}, where the peppers come out a dull orange).
As a result, an object lit unevenly is cut in two (\texttt{113044}, \texttt{12003}).
In contrast, the cylindrical methods keep each of these objects whole across the same range of illumination, because the hue of a surface is unchanged by how brightly it is lit.
The price is a visibly coarser reconstruction, since \rgbkmeans minimizes exactly the squared error between the quantized and the original image.
Moreover, where two distinct objects happen to share a hue, no cylindrical variant can separate them, as visible in image \texttt{299091} where the pyramid and the desert fall in one cluster.

\sengupta, finally, behaves as in the simulation study, where it is either the best
method or the worst. It gives the smoothest rendering of \texttt{113044}, the field
reduced to a single green against which the two animals stand out most sharply, and
it is on a par with the cylindrical $K$-means on \texttt{124084}, whose groups are
large and well populated. But it absorbs small, chromatically distinct regions into
the dominant component: the red dome of \texttt{118035} disappears into the white
walls of the church, and the peppers of \texttt{25098} break up. 
%Its Model~1 components share the angular concentration and the conditional linear variance, and describe the value as a sinusoidal function of the hue, so a small region with a distinct hue but an ordinary value is cheaper to explain as the tail of a wide component than as a component of its own.

\section{Application to wind direction and \texorpdfstring{\nox}{NOx} concentration}\label{sec:case1}

In this section, we analyze an application of the proposed methodology to the joint behavior of wind direction and pollutant concentration, which naturally have a cylindrical structure \citep{garcia2013exploring}: the direction is an angle, the
concentration a positive linear quantity, and the two are coupled by the geography of
the emission sources around the monitoring site.

\subsection{Data}\label{subsec:case1_data}

We use the \texttt{mydata} set distributed with the \texttt{openair} package
\citep{carslaw2012openair}: hourly measurements at the Marylebone Road kerbside
station in central London from 1 January 1998 to 23 June 2005.
It records wind direction and speed together with the concentrations of six
pollutants: nitrogen oxides (\nox), nitrogen dioxide alone (NO\textsubscript{2}), ozone (O\textsubscript{3}), particulate matter under $10\,\mu\mathrm{m}$ in diameter (PM\textsubscript{10}),
sulphur dioxide (SO\textsubscript{2}) and carbon monoxide (CO). We take the wind
direction as the angular coordinate and the \nox\ concentration, in parts per billion
by volume (ppb), as the linear one: it is the pollutant with the strongest
directional signal at this site (the $R^2$ of the concentration on
$(\cos\theta,\sin\theta)$ is $0.39$, against $0.12$--$0.34$ for the others).
Hourly readings are strongly serially dependent, so we retain one observation per day, the
12:00 reading; calms, whose direction is undefined, and incomplete records are
discarded. This leaves $n=2540$ daily observations. Wind direction is reported at a
$10^{\circ}$ resolution, $360^{\circ}$ being folded onto $0^{\circ}$; the \nox\
concentration is markedly right-skewed (median $188$, mean $210$, standard deviation
$134$, maximum $1092$ ppb).

Figure~\ref{fig:case1_data} shows the data through a cylindrical kernel density
estimate, the product of a von Mises kernel on the direction (concentration
$\kappa=15$) and a Gaussian kernel on the concentration, so that mass is not split at
the $0/360^{\circ}$ seam as it would be by a planar density. Two regimes are visible.
Winds from the south-western quadrant, roughly $180^{\circ}$--$270^{\circ}$, bring
high and widely dispersed concentrations, with the mode above $250$ ppb; winds from
around North bring low concentrations, with a tight mode below $100$ ppb. The second
regime spans both sides of the cut, from about $300^{\circ}$ to about $60^{\circ}$:
the unrolled view of panel~(a) splits it between its two edges, whereas the periodic
density and the top-down view of panel~(b) show it as a single mode.

\begin{figure}[!htpb]
  \centering
  \includegraphics[width=\textwidth]{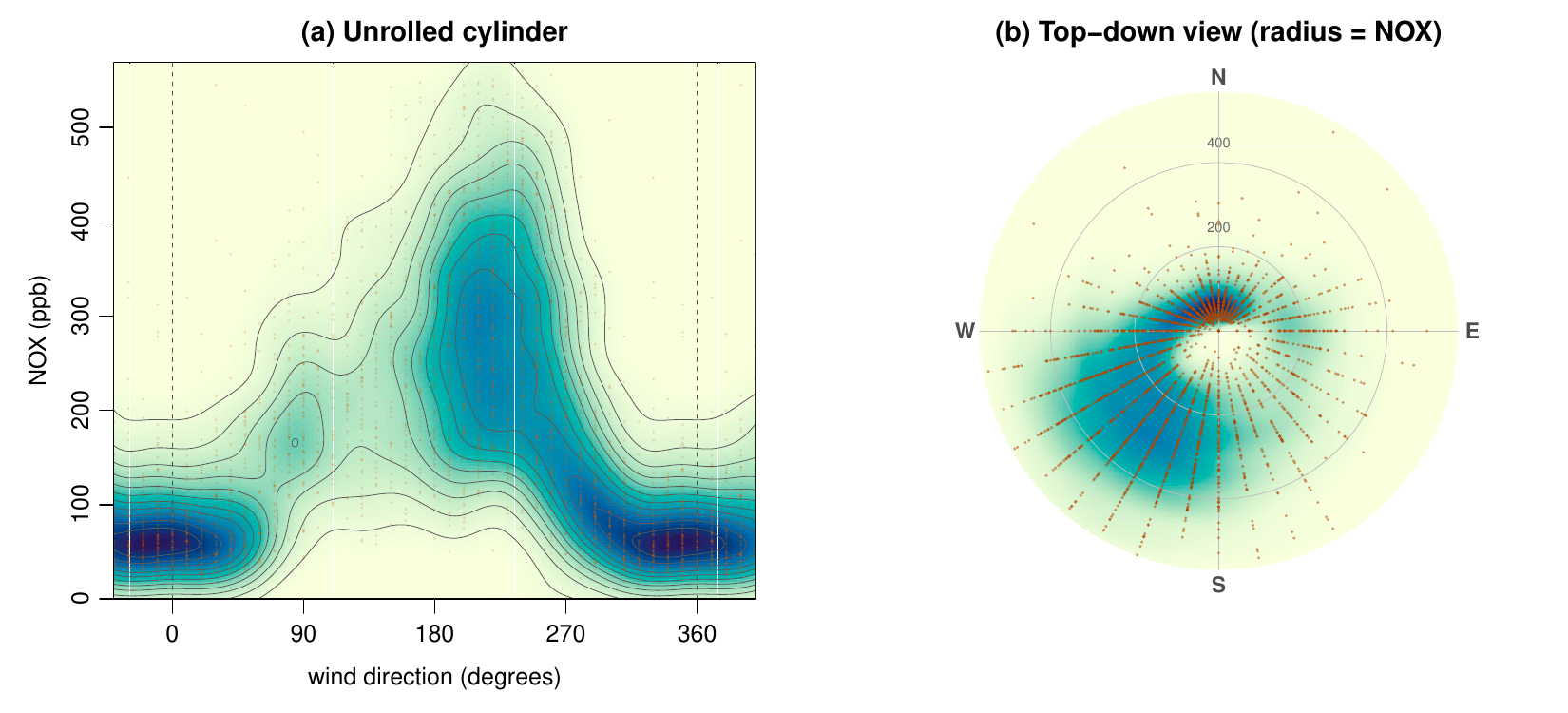}
  \caption{Wind direction and \nox\ concentration at Marylebone Road ($n=2540$ daily
    observations at 12:00), through the cylindrical (von Mises $\times$ Gaussian)
    kernel density estimate with the data overlaid. (a) Unrolled cylinder, shown just
    beyond one full turn so that the periodicity is visible; the dashed lines mark one
    period. (b) Top-down view, with the direction as compass bearing and the
    concentration as radius. The display is capped at the 99th percentile of \nox\
    ($569$ ppb).\label{fig:case1_data}}
\end{figure}

\subsection{Methodology}\label{subsec:case1_methods}

The two coordinates must be made commensurable before a distance can be computed on
them. We standardize the concentration by its sample mean $\bar z$ and standard
deviation $s_z$, and set the radius of the cylinder to one, so
that the linear coordinate is $x = (z-\bar z)/s_z$ and the squared geodesic distance
is $\big(\delta^{\mathrm{geo}}\big)^2 + (x_i-x_j)^2$; a different radius $r$ amounts to rescaling $x$ by
$1/r$.\footnote{The choice of $r$ is a modeling decision, not an estimate. We swept
$r$ over $\{0.25, 0.5, 1, 2, 4\}$, rescoring each $K$-means variant under the distance
it minimizes at that radius: the mean silhouette width is maximal at $K=2$ in eleven
of the twelve curves, so the choice of $K$ made below does not depend on the radius.
The partition itself does: as $r$ grows the linear coordinate shrinks and the split
turns from one by concentration into one by sector, abruptly for \euckmeans and
gradually for the two cylindrical variants. \sengupta, which estimates its own linear
variance, is invariant to $r$ by construction. We report $r=1$ throughout; the sweep
is available on request.} We compare the four
methods of Section~\ref{subsec:sim_methods}, with the same search budget as in
Section~\ref{sec:images}: $50$ $k$-means++ restarts and at most $300$ iterations for the $K$-means variants, $50$ random initializations and $300$ EM iterations for
\sengupta.

The number of clusters is dictated by the question: whether the wind separates the
site into a polluted and a clean regime, that is, $K=2$.
Nevertheless, the average silhouette width~\citep{rousseeuw1987silhouettes} is maximal at $K=2$ for the three $K$-means methods (each computed using the method's own distance) for
$K\in\{2,\ldots,6\}$. 
For \sengupta, $K$ is selected via the Bayesian Information Criterion (BIC,~\cite{schwarz1978estimating}), which here decreases monotonically over the same range and would therefore select $K=6$.
To have a fair comparison with the other methods, we also fix $K=2$ for \sengupta.
The diagnostics and the partitions at every $K$ are reported in
Appendix~\ref{sec:case1_diagnostics}. 

\subsection{Results}\label{subsec:case1_results}

Figure~\ref{fig:case1_clusters} shows the $K=2$ partitions, in the unrolled view and
as clustered pollution roses. \chordkmeans and \geokmeans agree almost exactly (ARI
$0.887$ between them) and recover the two regimes read off the density: a polluted
sector, centered at $212^{\circ}$ and holding $1500$ days with a mean concentration of
$284$ ppb, and a clean sector, centered at $357^{\circ}$ and holding $1040$ days with a
mean concentration of $103$ ppb. The clean sector straddles the cut and is returned as
one cluster.

\begin{figure}[!htpb]
  \centering
  \includegraphics[width=\textwidth]{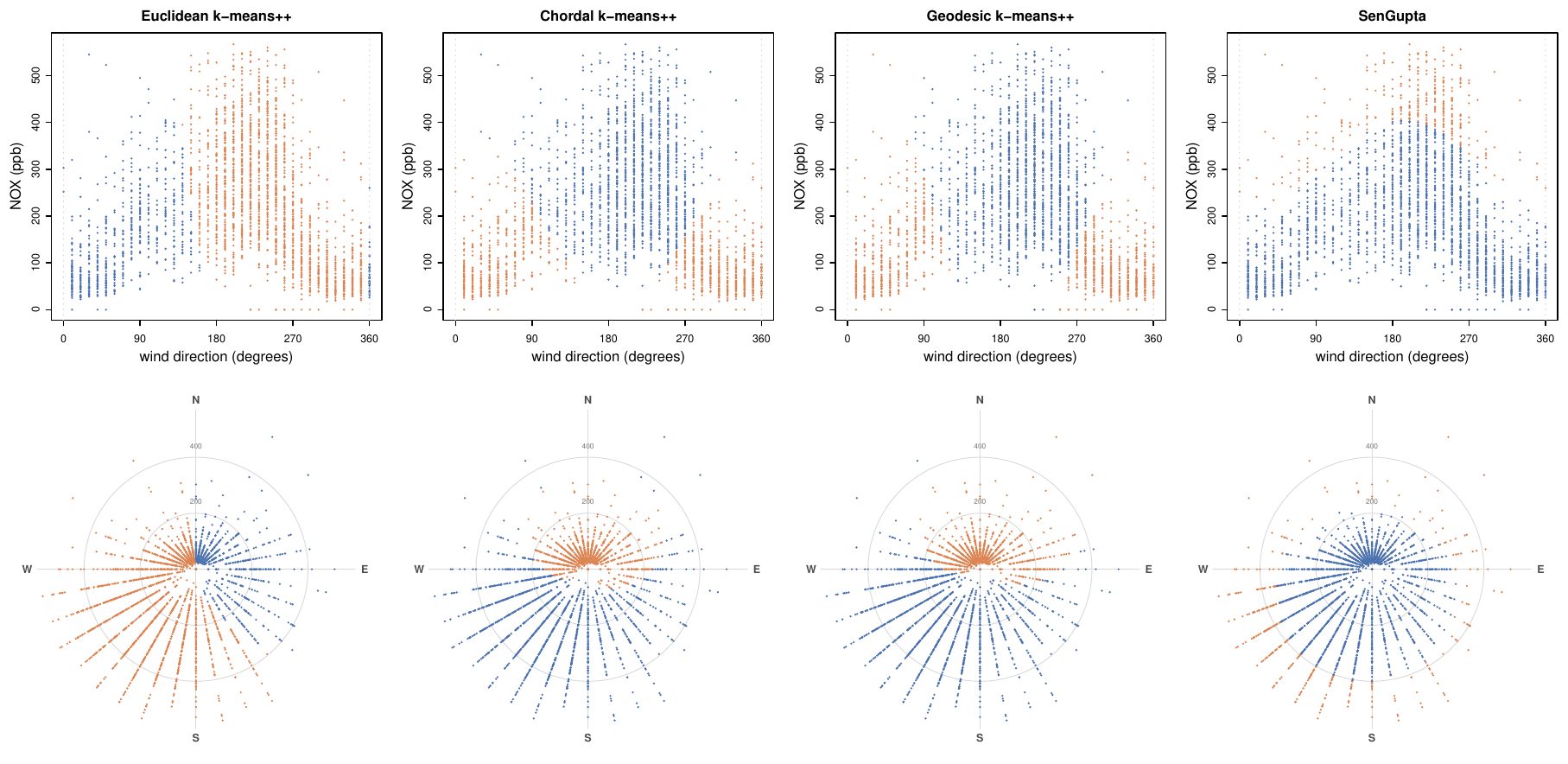}
  \caption{The $K=2$ partitions of the wind direction and \nox\ data, one method per
    column. Top: unrolled view, direction against concentration. Bottom: clustered
    pollution rose, with the direction as compass bearing and the concentration as
    radius. Colors mark the clusters; the display is capped at the 99th percentile of
    \nox.\label{fig:case1_clusters}}
\end{figure}

\euckmeans returns a different partition (ARI $0.21$ with \geokmeans). Unable to see
that $350^{\circ}$ and $10^{\circ}$ are neighbors, it places a boundary at the seam
and another at about $130^{\circ}$: its first cluster ($686$ days, centered at
$64^{\circ}$) collects the directions from North to East-South-East, and its second
($1854$ days) everything else, so that the north-westerly half of the clean regime is
lumped with the polluted south-westerly sector.

\sengupta returns yet another partition, unrelated to the other three (ARI below
$0.05$ with each). Its two components are not sectors but concentration levels: a
bulk component ($2197$ days, mean $174$ ppb, nearly uniform in direction) and a high-pollution component ($343$ days, mean $441$ ppb, centered at $219^{\circ}$).
The mixture is describing the right-skewed marginal
distribution of \nox\ rather than its dependence on the wind, and it keeps doing so
as $K$ grows: the partitions at $K=3,\ldots,6$ in Appendix~\ref{sec:case1_diagnostics}
are stacked concentration bands, which is what the monotonically decreasing BIC is
rewarding. On these data the model-based method answers a question about the shape of
the concentration distribution; the cylindrical $K$-means answer the one that was
asked.

\section{Conclusions}\label{sec:conclusions}

We have proposed $K$-means clustering for cylindrical data that respects the geometry
of $S^1\times\mathbb{R}$.
The construction combines the linear distance with either of the two classical
distances on the circle, the chord and the arc, and yields two algorithms,
\chordkmeans and \geokmeans, differing in the metric alone. The two are not
interchangeable: the chord saturates where the arc keeps growing (Remark~\ref{remark:arc_vs_chord}), so that they agree to second order when the within-cluster angular spread is small and part company when it is not.
For both we established the convergence of the algorithm (Section~\ref{subsec:convergence}) and its computational cost, the arc-length centroid being available exactly in $O(m\log m)$ per cluster through Algorithm~\ref{alg:frechet_circular}.
We then compared them against the naive Euclidean $K$-means and against the model-based
mixture of \cite{sengupta2020model} on ten simulated designs and on two applications.

The evidence shows that the structure cannot be neglected: \euckmeans
is never meaningfully better than the cylindrical variants and collapses as soon as a
cluster crosses the $0/2\pi$ junction. Between the two cylindrical variants the choice
is decided by whether the angle carries information. When it does, the geodesic
distance is the exact geometry and the natural default, while the chordal one
approximates it closely and more cheaply; when it does not, the saturation of the
chord acts as an unintended hedge against an uninformative coordinate, and
\chordkmeans wins for a reason that has nothing to do with geometry. 
The model-based clustering algorithm of \cite{sengupta2020model} is the
most variable method of the four, its EM search returning either the best or the worst
partition of a design.
We recommend using it whenever the clusters are expected to be elongated along
cluster-specific directions, the shapes its Mardia--Sutton components describe and that
no squared-distance criterion can follow, always comparing it to alternative
clustering methods.
The applications tell the same story: on color
quantization the cylindrical variants segment objects that the RGB baseline breaks
apart along the illumination, and on the wind direction and \nox\ data they recover
the two pollution regimes that \euckmeans cuts at the seam of the coordinate system.
The one question we leave open is the one that decides the balance between the two
coordinates. Their relative weight is not a property of the data but a modeling
choice --- in the applications it is the radius of the cylinder, and it is what makes
one coordinate dominate the partition --- and how to set it deserves a study of its
own.

% The CRediT and the funding statements name the authors: they are printed only in the
% version with author details (\anon = 1) and disappear from the anonymized one.
\if1\anon
  {
    \section{Author contributions}\label{sec:credit}

    CRediT: \textbf{Giuseppe Pandolfo}: Conceptualization, Methodology, Software, Validation, Formal analysis, Investigation, Data curation, Visualization, Writing -- original draft, Writing -- review \& editing. \textbf{Luca Coraggio}: Conceptualization, Methodology, Software, Validation, Formal analysis, Investigation, Data curation, Visualization, Writing -- original draft, Writing -- review \& editing. \textbf{Antonio D'Ambrosio}: Validation, Writing -- original draft, Writing -- review \& editing.

    \section{Funding}\label{sec:funding}

    Luca Coraggio acknowledges funding by the European Union -- \emph{NextGenerationEU}, in the framework of the \emph{GRINS -- Growing Resilient, INclusive and Sustainable} project (GRINS PE00000018 - CUP E63C22002140007). The views and opinions expressed are solely those of the authors and do not necessarily reflect those of the European Union, nor can the European Union be held responsible for them.
  } \fi

\section{Disclosure statement}\label{sec:disclosure}

The authors report there are no competing interests to declare.

\section{Declaration of generative AI use}\label{sec:ai_declaration}

% T&F policy: the statement must give "the full name of the tool used (with version
% number), how it was used and the reason for use". Complete the list of tools below.
The research question, the methodology, the design and the interpretation of the experiments and the original draft of the manuscript are the authors' own work. Generative AI tools were used as assistants in its preparation: Claude (Anthropic; model Claude Fable 5.1, Opus 5 and Sonnet 5). They were used for coding assistance in developing and checking the replication code and the scripts that produce the figures; for checking derivations; and for rewording and copyediting authors' own drafts and notes. No data, result or image analyzed in the paper was generated or altered by these tools: every numerical result is produced by the replication code, which the authors ran and verified. All AI-assisted text and code were rigorously reviewed and revised by the authors, who take full responsibility for the content of the article.

\section{Data Availability Statement}\label{sec:data_availability}

All three data sources used in this paper are public and none contains personal data.
The simulated designs of Section~\ref{sec:simulations} are generated by the
replication code from the Mardia--Sutton parameters of
Table~\ref{tab:simconfigs}, so that no data file needs to be distributed. The images
of Section~\ref{sec:images} are the six listed in Table~\ref{tab:images}, taken from
the Berkeley Segmentation Data Set BSDS300 \citep{martin2001database}, which is
distributed by the Computer Vision Group of the University of California, Berkeley.
The wind direction and \nox\ data of Section~\ref{sec:case1} are the \texttt{mydata}
set of the \texttt{openair} R package \citep{carslaw2012openair}, hourly measurements
at the Marylebone Road kerbside station in central London between 1 January 1998 and
23 June 2005, available from CRAN with the package and from its public repository.
The replication code obtains it automatically. Code reproducing every table and
figure of the paper, including the selection of the designs and of the images, is
provided as supplementary material to this submission and will be deposited in a
public repository upon acceptance.

\phantomsection\label{sec:supplementary}
\bigskip

\bibliography{bibliography}

\clearpage
\begin{center}
  {\large\bf SUPPLEMENTARY MATERIAL}
\end{center}

\begin{description}
  \item[Appendices:] Exact computation of the circular Fr\'echet mean, with proofs (Appendix~A); arc versus chord, a worked example (Appendix~B); full specification of the simulated designs (Appendix~C); model selection for the wind direction and \nox\ data (Appendix~D). They are reported below for the convenience of the reviewers. (pdf file)
  \item[Replication code and data:] R code reproducing every table and figure of the paper, together with the specification of the simulated designs, the list of the selected BSDS300 images and the scripts that retrieve the public data sets; a README file describes the workflow. (zip file)
\end{description}

\appendix
% statements are numbered sequentially in the body and per-section in the
% appendices, so one can be promoted or demoted without renumbering by hand
\counterwithin{proposition}{section}
\counterwithin{lemma}{section}

\section{Exact computation of the circular Fr\'echet mean}\label{sec:frechet}

This appendix gives a self-contained derivation of the angular centroid~\eqref{eq:centroid_geo}, the Fr\'echet mean of the arc length on the circle (the \emph{intrinsic mean} of the directional statistics literature). \citet{McKilliam2012} compute it in linear time by a different route, reducing the minimization to a nearest-lattice-point problem; \citet{HotzHuckemann2015} show that it lies among at most \(m\) candidates, the arithmetic means of the sample unwrapped in each of the ways compatible with a sorted order, and \citet{Cazals2021} extend the same construction to \(p\)-means and to exact arithmetic. Algorithm~\ref{alg:frechet_circular} is a different writing of the sweep over candidates of the last two works, expressed through the wrapping map introduced below; it computes the same quantity, and we state and prove its correctness here so that the paper can be read on its own.

\subsection{The wrapping map and the three-candidate rule}\label{subsec:wr}

Let \(\{\theta_i\}_{i=1}^m\) denote the angular coordinates of the points in a given cluster. For \(a \in \mathbb{R}\) define the \emph{wrapping map}
\[
\mathrm{wr}(a) = \big( (a + \pi) \bmod 2\pi \big) - \pi \in [-\pi, \pi),
\]
where ``\(x \bmod n\)'' denotes the unique representative of \(x\) in \([0, n)\), which is non-negative whatever the sign of \(x\). Whereas \(\delta^{\mathrm{geo}}\) returns the length of the shortest arc, \(\mathrm{wr}\) returns the \emph{signed} angular displacement that realizes it. The following lemma reconciles the two descriptions of the minimal angular difference used in the paper --- the minimization over \(\varphi \in \{-1,0,1\}\) of Section~\ref{subsec:geodesic_distance} and the wrapping map. The reconciliation is not a formality: the piecewise structure of the objective is a direct consequence of it, since it is the switching of the minimizing \(\varphi\), and nothing else, that creates the breakpoints.

\begin{lemma}\label{lemma:wr}
  For every \(a \in \mathbb{R}\),
\[
  \min_{\varphi \in \mathbb{Z}} |a + 2\pi\varphi| \;=\; |\mathrm{wr}(a)| ,
\]
the minimum being attained at the unique \(\varphi\) for which \(a + 2\pi\varphi \in [-\pi,\pi)\), except when \(\mathrm{wr}(a) = -\pi\), in which case the two representatives \(\pm\pi\) both attain it. If in addition \(a = \theta - \vartheta\) with \(\theta, \vartheta \in [0,2\pi)\), the minimizing \(\varphi\) belongs to \(\{-1,0,1\}\), so that
\[
  \delta^{\mathrm{geo}}(\theta,\vartheta)
  \;=\; \min_{\varphi \in \{-1,0,1\}} |\theta - \vartheta + 2\pi\varphi|
  \;=\; |\mathrm{wr}(\theta-\vartheta)| .
\]
\end{lemma}

\noindent\textbf{Proof.} The candidates \(\{a + 2\pi\varphi : \varphi \in \mathbb{Z}\}\) form a lattice of spacing \(2\pi\). The intervals \(\big[(2k-1)\pi,\,(2k+1)\pi\big)\), \(k \in \mathbb{Z}\), tile \(\mathbb{R}\) without overlap, so exactly one lattice point falls in \([-\pi,\pi)\); by construction of the modulo operation that point is \(\mathrm{wr}(a)\). Any other lattice point is at distance at least \(2\pi\) from it, hence has modulus at least \(2\pi - |\mathrm{wr}(a)| \ge \pi \ge |\mathrm{wr}(a)|\). Both inequalities are strict unless \(|\mathrm{wr}(a)| = \pi\), that is unless \(\mathrm{wr}(a) = -\pi\), in which case \(-\pi\) and \(+\pi\) are two distinct minimizers.

For the second statement, \(\theta,\vartheta \in [0,2\pi)\) forces \(a \in (-2\pi, 2\pi)\), and the integer that brings \(a\) into \([-\pi,\pi)\) is \(\varphi = 0\) if \(a \in [-\pi,\pi)\), \(\varphi = -1\) if \(a \in [\pi, 2\pi)\), and \(\varphi = +1\) if \(a \in (-2\pi,-\pi)\). Restricting the search to \(\{-1,0,1\}\) is therefore an exhaustive enumeration. \qed

\medskip
Intuitively, $\mathrm{wr}(\cdot)$ induces a representation of the angle, $\theta_i$, that only depends on its antipode, $a_i=(\theta_i + \pi) \bmod 2\pi$, and that remains constant until the latter is crossed.
To see this, fix \(\theta_i\) and read \(\mathrm{wr}(\theta_i - \vartheta)\) as a function of \(\vartheta\), writing \(\mathrm{wr}(\theta_i - \vartheta) = \theta_i - \vartheta + 2\pi\varphi_i(\vartheta)\) with \(\varphi_i(\vartheta)\) the unique integer of the Lemma; $\theta_i + 2\pi\varphi_i(\vartheta)$ is the representative. 
As \(\vartheta\) increases, \(\theta_i - \vartheta\) decreases with slope \(-1\) and leaves \([-\pi,\pi)\) through its left endpoint exactly when \(\theta_i - \vartheta = -\pi\), that is when \(\vartheta = \theta_i + \pi \equiv a_i \pmod{2\pi}\). At that value \(\varphi_i\) increases by one and \(\mathrm{wr}(\theta_i - \vartheta)\) jumps from \(-\pi\) to \(+\pi\). Between two consecutive cut points \(\varphi_i\) is thus \emph{constant}.
The same holds for the entire vector of angles: the representatives stay constant between any two consecutive antipodes.

Moreover, since the jump of \(\mathrm{wr}\) is from \(-\pi\) to \(+\pi\), both \(|\mathrm{wr}|\) and  \(\mathrm{wr}^2\) are continuous everywhere: at \(\vartheta = a_i\) the \(i\)-th summand equals \(\pi^2\) on both sides.

Both facts are used below.

\subsection{Piecewise structure of the objective and the algorithm}\label{subsec:proofprop1}

By Lemma~\ref{lemma:wr}, the angular objective can be written as
\[
F(\vartheta) = \sum_{i=1}^m \delta^{\mathrm{geo}}(\theta_i, \vartheta)^2 = \sum_{i=1}^m \mathrm{wr}(\theta_i - \vartheta)^2 .
\]
Define the \emph{cut points} \(a_i = (\theta_i + \pi) \bmod 2\pi\), that is, the antipodes of the observations, and let \(a_{(1)} < \cdots < a_{(m')}\) be their distinct values sorted increasingly, with the convention \(a_{(m'+1)} = a_{(1)} + 2\pi\). They partition the circle into the arcs \(\mathcal{I}_j = \big(a_{(j)}, a_{(j+1)}\big)\), \(j = 1, \ldots, m'\).

\begin{proposition}\label{prop:frechet}
  The function \(F\) is continuous on \(S^1\) and piecewise quadratic, its breakpoints being contained in the cut set \(\mathcal{A} = \{a_i\}_{i=1}^m\). On each closed arc \(\overline{\mathcal{I}_j}\) it coincides with the convex quadratic \(\vartheta \mapsto \sum_{i=1}^m \big(\theta_i^{(j)} - \vartheta\big)^2\), where \(\theta_i^{(j)} = c_j + \mathrm{wr}(\theta_i - c_j)\) for any interior point \(c_j \in \mathcal{I}_j\). Consequently, Algorithm~\ref{alg:frechet_circular} returns a global minimizer of \(F\).
\end{proposition}

\noindent\textbf{Proof.} Let \(F(\vartheta) = \sum_{i=1}^m \delta^{\mathrm{geo}}(\theta_i,\vartheta)^2\), let \(a_{(1)} < \cdots < a_{(m')}\) be the distinct cut points with \(a_{(m'+1)} = a_{(1)} + 2\pi\), and let \(\mathcal{I}_j = \big(a_{(j)}, a_{(j+1)}\big)\).

\smallskip
\noindent\emph{Step 1: the summands are quadratic on each arc.} By Lemma~\ref{lemma:wr}, \(\delta^{\mathrm{geo}}(\theta_i,\vartheta)^2 = \mathrm{wr}(\theta_i - \vartheta)^2\). 
By the previous observation, on the open arc \(\mathcal{I}_j\) there exist integers \(\varphi_i^{(j)}\), independent of \(\vartheta\), such that
\[
\mathrm{wr}(\theta_i - \vartheta) = \theta_i^{(j)} - \vartheta ,
\qquad \theta_i^{(j)} = \theta_i + 2\pi\varphi_i^{(j)} ,
\qquad \vartheta \in \mathcal{I}_j .
\]
Evaluating this identity at any interior point \(c_j \in \mathcal{I}_j\) and solving for \(\theta_i^{(j)}\) gives the constructive formula \(\theta_i^{(j)} = c_j + \mathrm{wr}(\theta_i - c_j)\) used in Algorithm~\ref{alg:frechet_circular}. Consequently \(F(\vartheta) = \sum_{i=1}^m \big(\theta_i^{(j)} - \vartheta\big)^2\) on \(\mathcal{I}_j\) is a convex quadratic function in \(\vartheta\).

\smallskip
\noindent\emph{Step 2: extension to the closed arc.} \(F\) is continuous on \(S^1\) by previous observations, and $\mathrm{wr}(\theta_i-\vartheta)^2$ is continuous on the closure of the $j\text{-th}$ interval, \(\overline{\mathcal{I}_j}\); two continuous functions that agree on the dense subset \(\mathcal{I}_j\) agree on its closure. The identity therefore holds on \(\overline{\mathcal{I}_j}\), which also shows that the breakpoints of \(F\) are contained in \(\mathcal{A} = \{a_i\}\).

\smallskip
\noindent\emph{Step 3: minimizing over one arc.} The unconstrained vertex of that quadratic is the arithmetic mean of the representatives, \(\bar m_j = m^{-1}\sum_{i=1}^m \theta_i^{(j)}\). A convex function restricted to a compact interval attains its minimum at the unconstrained vertex if that vertex lies inside, and at the nearer endpoint otherwise; both cases are covered by the clamp \(\mu_j = \min\big(\max(\bar m_j, a_{(j)}),\, a_{(j+1)}\big)\).

\smallskip
\noindent\emph{Step 4: from local to global.} The closed arcs \(\overline{\mathcal{I}_1}, \ldots, \overline{\mathcal{I}_{m'}}\) cover \(S^1\), so \(\min_{S^1} F = \min_j \min_{\overline{\mathcal{I}_j}} F = \min_j F(\mu_j)\), and the arc attaining the smallest \(F(\mu_j)\) contains a global minimizer. This is exactly what Algorithm~\ref{alg:frechet_circular} returns. \qed

\begin{algorithm}[htbp]
  \begin{spacing}{0.8}
    \caption{Exact Fr\'echet mean on the circle}
    \label{alg:frechet_circular}
    \begin{algorithmic}[1]
      \Require Angles \(\{\theta_i\}_{i=1}^{m} \subset [0, 2\pi)\)
      \Ensure \(\bar{\theta} \in \arg\min_{\vartheta \in [0,2\pi)} \sum_{i=1}^m \delta^{\mathrm{geo}}(\theta_i, \vartheta)^2\)
      \State Let \(\mathrm{wr}(a) = \big( (a + \pi) \bmod 2\pi \big) - \pi\)
      \State Compute the cut points \(a_i = (\theta_i + \pi) \bmod 2\pi\); let \(a_{(1)} < \cdots < a_{(m')}\) be their distinct sorted values and set \(a_{(m'+1)} = a_{(1)} + 2\pi\)
      \State \(F^\star \gets +\infty\)
      \For{\(j = 1\) to \(m'\)}
      \State \(c_j \gets \tfrac{1}{2}\big(a_{(j)} + a_{(j+1)}\big)\) \Comment{any interior point of \(\mathcal{I}_j\)}
      \State \(\theta_i^{(j)} \gets c_j + \mathrm{wr}(\theta_i - c_j)\), \ \(i = 1, \ldots, m\) \Comment{unwrapped representatives}
      \State \(\bar{m}_j \gets \frac{1}{m} \sum_{i=1}^m \theta_i^{(j)}\) \Comment{vertex of the quadratic on \(\mathcal{I}_j\)}
      \State \(\mu_j \gets \min\big( \max(\bar{m}_j, a_{(j)}), a_{(j+1)} \big) \bmod 2\pi\) \Comment{clamp to \(\overline{\mathcal{I}_j}\)}
      \State \(F_j \gets \sum_{i=1}^m \mathrm{wr}(\theta_i - \mu_j)^2\)
      \If{\(F_j < F^\star\)}
      \State \(F^\star \gets F_j\); \quad \(\bar{\theta} \gets \mu_j\)
      \EndIf
      \EndFor
      \State \Return \(\bar{\theta}\)
    \end{algorithmic}
  \end{spacing}
\end{algorithm}

\medskip
Two properties of the algorithm follow from the proof. First, the loop cannot be shortened to the single arc that happens to contain the observations: the global minimizer may sit on a different arc, and Figure~\ref{fig:frechet_arcs} exhibits such a case. Second, moving from \(\mathcal{I}_j\) to \(\mathcal{I}_{j+1}\) crosses exactly one cut point, so exactly one representative changes, by \(\pm 2\pi\); maintaining the two running totals \(S_j = \sum_i \theta_i^{(j)}\) and \(Q_j = \sum_i \big(\theta_i^{(j)}\big)^2\) therefore updates in \(O(1)\) both the vertex \(\bar m_j = S_j/m\) and the value \(F_j = Q_j - 2\mu_j S_j + m\mu_j^2\) that selects the winning arc, so that the body of the loop no longer touches the observations one by one. Since the cut points are sorted once in any case, the whole loop costs \(O(m \log m)\) instead of \(O(m^2)\).

\begin{figure}[htbp]
  \centering
  \begin{tikzpicture}[scale=1,>=Latex]
    \def\R{2.5}
    \draw[line width=2.4pt, black]                                (70:\R)  arc (70:200:\R);
    \draw[line width=2.4pt, black!45, dash pattern=on 5pt off 3pt] (200:\R) arc (200:280:\R);
    \draw[line width=2.4pt, black!45, dotted]                      (280:\R) arc (280:430:\R);
    \draw[black!25] (0,0) circle (\R);
    \fill (0,0) circle (1.3pt);
    \foreach \a/\lab in {20/{\theta_1}, 100/{\theta_2}, 250/{\theta_3}} {
      \draw[black!30, dashed] (\a:\R) -- ({\a+180}:\R);
      \fill (\a:\R) circle (2.7pt);
      \node at (\a:\R+0.40) {\(\lab\)};
    }
    \foreach \a/\lab in {200/{a_1}, 280/{a_2}, 70/{a_3}} {
      \draw[fill=white, line width=0.9pt] (\a:\R) circle (2.7pt);
      \node at (\a:\R+0.42) {\(\lab\)};
    }
    \draw[->, line width=1.4pt] (0,0) -- (3.33:\R);
    \node[fill=white, inner sep=1.5pt] at (3.33:1.35) {\(\bar\theta\)};
  \end{tikzpicture}

  \medskip
  {\small
  \begin{tabular}{@{}clcccrr@{}}
    \toprule
    & arc \(\mathcal{I}_j\) & \(\theta_1^{(j)}\) & \(\theta_2^{(j)}\) & \(\theta_3^{(j)}\) & \(\bar m_j\) & \(F(\mu_j)\)\\
    \midrule
    \swA & \((70^\circ, 200^\circ)\)  & \(20^\circ\)  & \(100^\circ\) & \(250^\circ\) & \(123.33^\circ\) & \(8.306\)\\
    \swB & \((200^\circ, 280^\circ)\) & \(380^\circ\) & \(100^\circ\) & \(250^\circ\) & \(243.33^\circ\) & \(11.961\)\\
    \swC & \((280^\circ, 430^\circ)\) & \(380^\circ\) & \(460^\circ\) & \(250^\circ\) & \(363.33^\circ\) & \(\mathbf{6.844}\)\\
    \bottomrule
  \end{tabular}}
  \caption{The piecewise structure exploited by Algorithm~\ref{alg:frechet_circular}, for the three angles \(\theta_1 = 20^\circ\), \(\theta_2 = 100^\circ\), \(\theta_3 = 250^\circ\) (filled dots). Each dashed diameter joins an observation to its own antipode \(a_i = (\theta_i + \pi) \bmod 2\pi\) (open dots), and it is the antipodes that cut the circle into three arcs, drawn in three line styles. Sorting them anticlockwise from \(\theta = 0\) renames them, which is the numbering used in Proposition~\ref{prop:frechet} and in the table below: \(a_{(1)} = a_3 = 70^\circ\), \(a_{(2)} = a_1 = 200^\circ\), \(a_{(3)} = a_2 = 280^\circ\). Within each arc the unwrapped representatives \(\theta_i^{(j)}\) are constant --- the table lists them --- and \(F\) coincides with a convex quadratic whose vertex is their arithmetic mean \(\bar m_j\); crossing a cut point moves exactly one representative by \(\pm 360^\circ\). Here every \(\bar m_j\) falls inside its own arc, so no clamping occurs, and the global minimizer is \(\bar\theta = 363.33^\circ \equiv 3.33^\circ\) (arrow), which lies on the dotted arc and not on the one containing the three observations. Inspecting that arc alone would have returned \(123.33^\circ\), at which \(F\) is \(21\%\) larger.}
  \label{fig:frechet_arcs}
\end{figure}
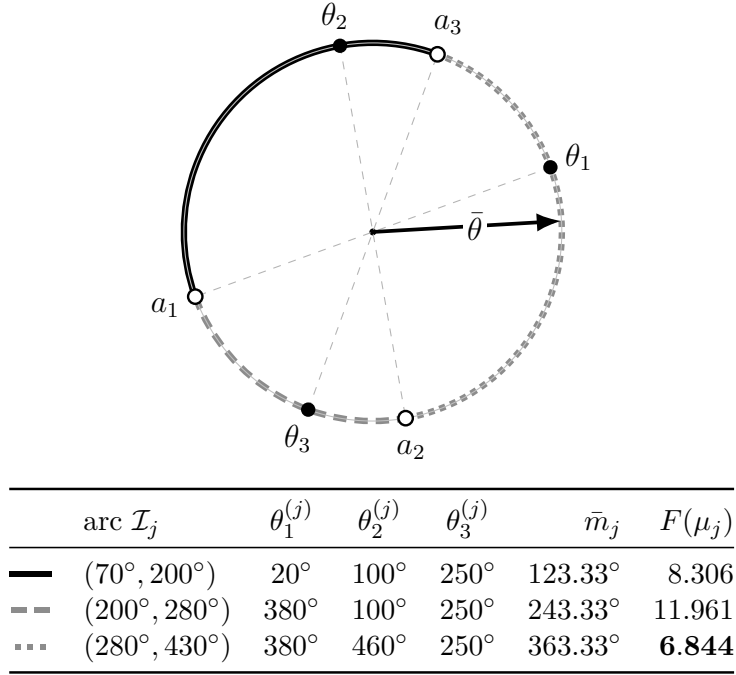

\section{Arc versus chord}\label{sec:arc_chord}

\begin{figure}[!htb]
  \centering
  \begin{subfigure}[b]{0.55\textwidth}
    \centering
    \begin{tikzpicture}
      \begin{axis}[
        width=0.96\linewidth, height=6cm,
        xmin=0, xmax=180, ymin=0, ymax=11,
        xtick={0,45,90,135,180},
        xticklabel={\(\pgfmathprintnumber{\tick}^\circ\)},
        ytick={0,2,4,6,8,10},
        xlabel={\(\Delta\)}, ylabel={squared angular cost},
        xlabel near ticks, ylabel near ticks,
        tick label style={font=\small}, label style={font=\small},
        legend cell align=left,
        legend style={at={(0.04,0.80)}, anchor=north west, draw=none, fill=none, font=\small},
        domain=0:180, samples=181,
        every axis plot/.append style={line width=1.2pt},
      ]
        \addplot[black] {(x*pi/180)^2};
        \addlegendentry{\(\Delta^2\) (arc)}
        \addplot[black!55, dashed] {4*sin(x/2)^2};
        \addlegendentry{\(4\sin^2(\Delta/2)\) (chord)}
        \draw[black!35, densely dotted] (axis cs:0,4) -- (axis cs:180,4);
        \draw[black!35, densely dotted] (axis cs:0,9.8696) -- (axis cs:180,9.8696);
        \node[anchor=south west, font=\scriptsize, black!60] at (axis cs:0,4) {\(4\)};
        \node[anchor=north west, font=\scriptsize, black!60] at (axis cs:0,9.8696) {\(\pi^2\)};
      \end{axis}
    \end{tikzpicture}
    \caption{squared arc and squared chord}
  \end{subfigure}\hfill
  \begin{subfigure}[b]{0.42\textwidth}
    \centering
    \begin{tikzpicture}[scale=0.85,>=Latex]
      \def\R{2.2}
      \draw[black!25] (0,0) circle (\R);
      \fill (0:\R) circle (3.0pt);   \node[anchor=north west] at (0:\R+0.10) {\small\(0^\circ (\times 3)\)};
      \fill (135:\R) circle (2.7pt); \node[anchor=south east] at (135:\R+0.06) {\small\(135^\circ\)};
      \fill (180:\R) circle (2.7pt); \node[anchor=east] at (180:\R+0.14) {\small\(180^\circ\)};
      \draw[->, line width=1.5pt] (0,0) -- (63:\R);
      \draw[->, line width=1.5pt, black!50, dash pattern=on 4pt off 2.5pt] (0,0) -- (28.68:\R);
      \node[anchor=south west] at (63:\R+0.06)    {\small\(\bar\theta_{\mathrm{geo}}\)};
      \node[anchor=west]       at (28.68:\R+0.10) {\small\(\bar\theta_{\mathrm{chord}}\)};
      \fill (0,0) circle (1.3pt);
    \end{tikzpicture}
    \caption{five angles, the two means}
  \end{subfigure}
  \caption{(a) Squared arc \(\Delta^2\) (solid) and squared chord \(4\sin^2(\Delta/2)\) (dashed) against the angular separation \(\Delta\): they agree near \(0\) and part company beyond, the chord saturating at \(4\) where the arc reaches \(\pi^2\). (b) The five angles \((0^\circ, 0^\circ, 0^\circ, 135^\circ, 180^\circ)\) with their mean direction \(\bar\theta_{\mathrm{chord}} = 28.68^\circ\) (dashed arrow) and their Fr\'echet mean \(\bar\theta_{\mathrm{geo}} = 63.00^\circ\) (solid arrow), \(34.32^\circ\) apart.}
  \label{fig:arc_chord}
\end{figure}
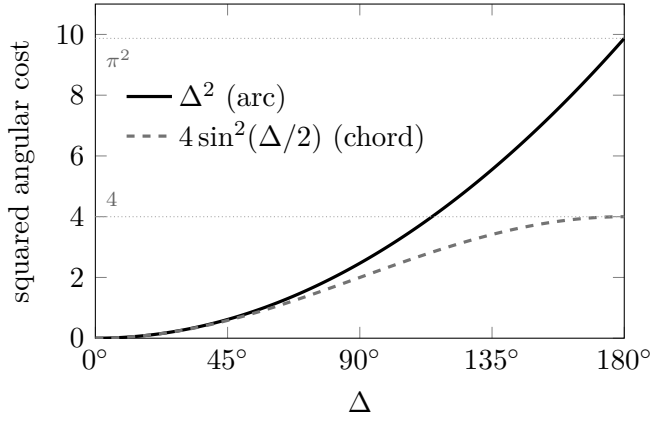
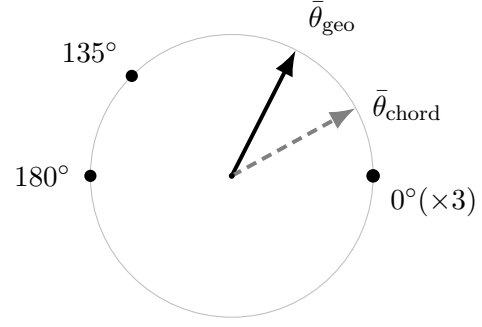

The two angular costs, \(\big(\delta^{\mathrm{geo}}\big)^2\) and \(\big(\delta^{\mathrm{chord}}\big)^2 = 4\sin^2(\delta^{\mathrm{geo}}/2) = 2(1 - \cos\delta^{\mathrm{geo}})\), are the squared \emph{arc} and the squared \emph{chord} subtending the same angular separation \(\Delta = \delta^{\mathrm{geo}} \in [0,\pi]\); Figure~\ref{fig:arc_chord}(a) plots both. Remark~1 records that they agree to second order; beyond that they diverge steadily: the ratio \(\psi(\Delta) = \Delta^2 / 4\sin^2(\Delta/2)\) increases monotonically from \(1\) as \(\Delta \to 0\) to \(\pi^2/4 \simeq 2.47\) at \(\Delta = \pi\), where the chord saturates at \(4\) while the arc reaches \(\pi^2 \simeq 9.87\). A point diametrically opposite the centroid thus weighs about two and a half times less under \(d_{\mathrm{chord}}\) than under \(d_{\mathrm{geo}}\), and the two centroids move apart accordingly.

\paragraph*{A worked example.}
Take the five angles \(\theta = (0^\circ, 0^\circ, 0^\circ, 135^\circ, 180^\circ)\) of Figure~\ref{fig:arc_chord}(b): a majority of three coincident observations and two stragglers pulling the other way. Each mean is optimal for its own cost and strictly suboptimal for the other's:

\begin{center}
  \begin{spacing}{1.0}
  \small
  \begin{tabular}{@{}lccc@{}}
    \toprule
    & \(\vartheta\) & \(F(\vartheta) = \sum_i \delta^{\mathrm{geo}}(\theta_i,\vartheta)^2\) & \(G(\vartheta) = \sum_i 2(1 - \cos(\theta_i - \vartheta))\)\\
    \midrule
    mean direction, \(\mathrm{atan2}\)                        & \(28.68^\circ\) & \(11.171\) & \(\mathbf{7.053}\)\\
    Fr\'echet mean, Algorithm~\ref{alg:frechet_circular}      & \(63.00^\circ\) & \(\mathbf{9.376}\)  & \(7.566\)\\
    \bottomrule
  \end{tabular}
  \end{spacing}
\end{center}

\noindent The two locations differ by \(34.32^\circ\): the mean direction sits \(19.1\%\) above the minimum of \(F\), the Fr\'echet mean \(7.3\%\) above the minimum of \(G\). The asymmetry is Remark~1 at work: the chord under-weights the far points, so the mean direction is dragged towards the majority at \(0^\circ\), whereas the Fr\'echet mean, charging \(135^\circ\) and \(180^\circ\) their full arc length, settles much closer to them.

\section{Full specification of the simulated designs}

Table~\ref{tab:sim-full} expands Table~\ref{tab:simconfigs} into one row per design,
reporting the parameters of every Mardia--Sutton component. Angular means are given in
degrees for readability. The eight factorial designs are reported at $\rho_1=\rho_2=0$;
each of them was also generated at the two positive correlation levels $\rho=0.5$ and
$\rho=0.9$, obtained by setting $\rho_{1,k}=\rho_{2,k}=\rho/\sqrt{2}$ for every
component and leaving all remaining parameters unchanged. As discussed in
Section~\ref{subsec:sim_results}, those replicates leave the comparison qualitatively
unchanged and are not reported.

\begin{table}[htbp]
  \centering
  \caption{Parameters of every simulated design. $\mu_0$ is the angular mean (degrees),
    $\mu$ the unconditional linear mean, $\kappa$ the angular concentration,
    $\sigma^2$ the unconditional linear variance and $(\rho_1,\rho_2)$ the
    angular--linear correlation parameters of each component. All designs use
    $n_k=100$ observations per cluster and $100$ Monte Carlo
    replications.}\label{tab:sim-full}
  \small
  \begin{tabular}{@{}llcccp{0.26\textwidth}@{}}
    \toprule
    Design & $\mu_0$ (deg) & $\mu$ & $\kappa$ & $\sigma^2$ & $(\rho_1,\rho_2)$\\
    \midrule
    \texttt{LL3} & $157.5,\ 112.5,\ 67.5$ & $0,2,4$ & $10$ & $0.1$ & $(0,0)$\\
    \texttt{HL3} & $157.5,\ 112.5,\ 67.5$ & $0,2,4$ & $0.5$ & $0.1$ & $(0,0)$\\
    \texttt{LH3} & $157.5,\ 112.5,\ 67.5$ & $0,2,4$ & $10$ & $2$ & $(0,0)$\\
    \texttt{HH3} & $157.5,\ 112.5,\ 67.5$ & $0,2,4$ & $0.5$ & $2$ & $(0,0)$\\
    \midrule
    \texttt{LL6} & $195,135,75,15,315,255$ & $0,2,4,4,2,0$ & $10$ & $0.1$ & $(0,0)$\\
    \texttt{HL6} & $195,135,75,15,315,255$ & $0,2,4,4,2,0$ & $0.5$ & $0.1$ & $(0,0)$\\
    \texttt{LH6} & $195,135,75,15,315,255$ & $0,2,4,4,2,0$ & $10$ & $2$ & $(0,0)$\\
    \texttt{HH6} & $195,135,75,15,315,255$ & $0,2,4,4,2,0$ & $0.5$ & $2$ & $(0,0)$\\
    \midrule
    \texttt{COR3} & $0,\ 60,\ 330$ & $0,0.5,0.5$ & $10$ & $0.5$ &
      $(0.60,0.75)$, $(-0.68,-0.68)$,\newline $(-0.75,-0.60)$\\
    \addlinespace
    \texttt{COR6} & $0,60,120,180,240,300$ & $0,1,0,1,0,1$ & $8$ & $0.6$ &
      $(0.60,0.75)$, $(0.68,0.68)$,\newline $(0.75,0.60)$, cycled\\
    \bottomrule
  \end{tabular}
\end{table}

\section{Model selection for the wind direction and \texorpdfstring{\nox}{NOx} data}\label{sec:case1_diagnostics}

Figure~\ref{fig:case1_diagnostics} reports the criteria used in
Section~\ref{subsec:case1_methods} to check the choice $K=2$, each method being scored
on its own terms. Panel~(a) shows, for $K=2,\ldots,6$, the mean silhouette width of
the partition returned by each $K$-means variant, computed under the distance that the
variant minimizes: the raw Euclidean distance on $(\theta,x)$ for \euckmeans, the
chordal distance~\eqref{eq:dchord} for \chordkmeans and the geodesic
distance~\eqref{eq:dgeo} for \geokmeans. The three curves are maximal at $K=2$ and
decrease from there on. Panel~(b) shows the BIC of \sengupta, which decreases
monotonically over the range and would select $K=6$, the largest value considered.
Figure~\ref{fig:case1_maps} shows the partitions behind these numbers. From $K=3$ on
the $K$-means variants subdivide the polluted sector by concentration level, while
\sengupta slices the whole range of directions into concentration bands at every $K$:
each additional component refines the description of the skewed \nox\ marginal, which
is what the BIC rewards.

\begin{figure}[!htpb]
  \centering
  \includegraphics[width=\textwidth]{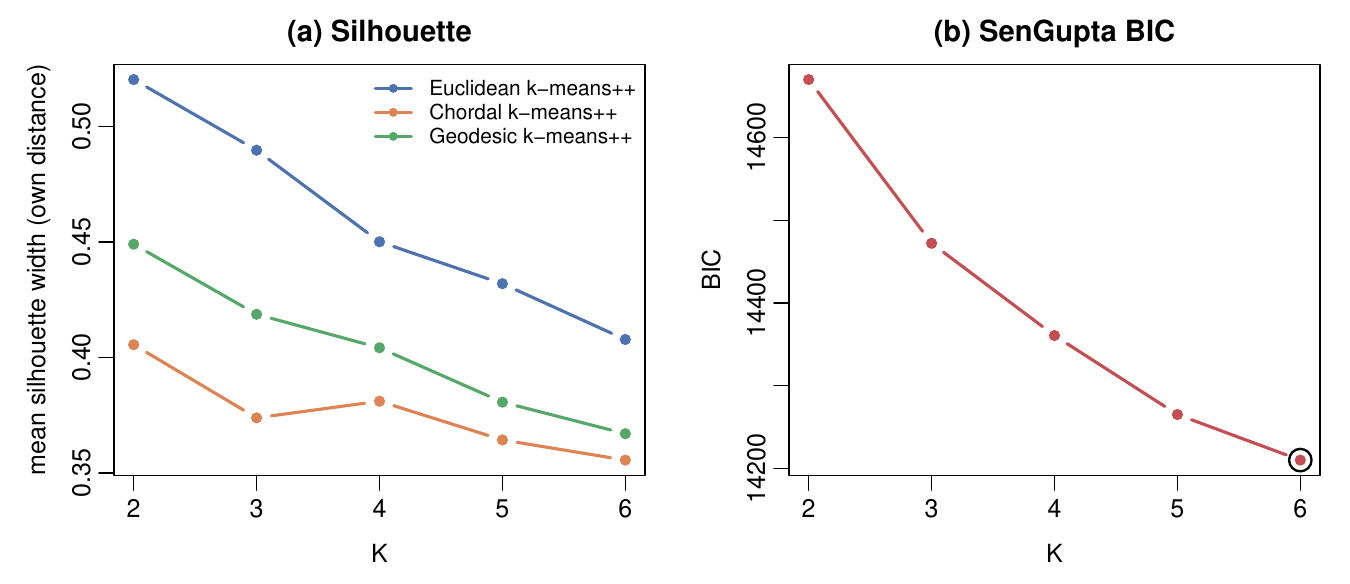}
  \caption{Choice of $K$ for the wind direction and \nox\ data. (a) Mean silhouette
    width of the partition of each $K$-means variant, computed under the variant's own
    distance. (b) BIC of \sengupta; the circle marks the minimum.\label{fig:case1_diagnostics}}
\end{figure}

\begin{figure}[!htpb]
  \centering
  \includegraphics[width=\textwidth]{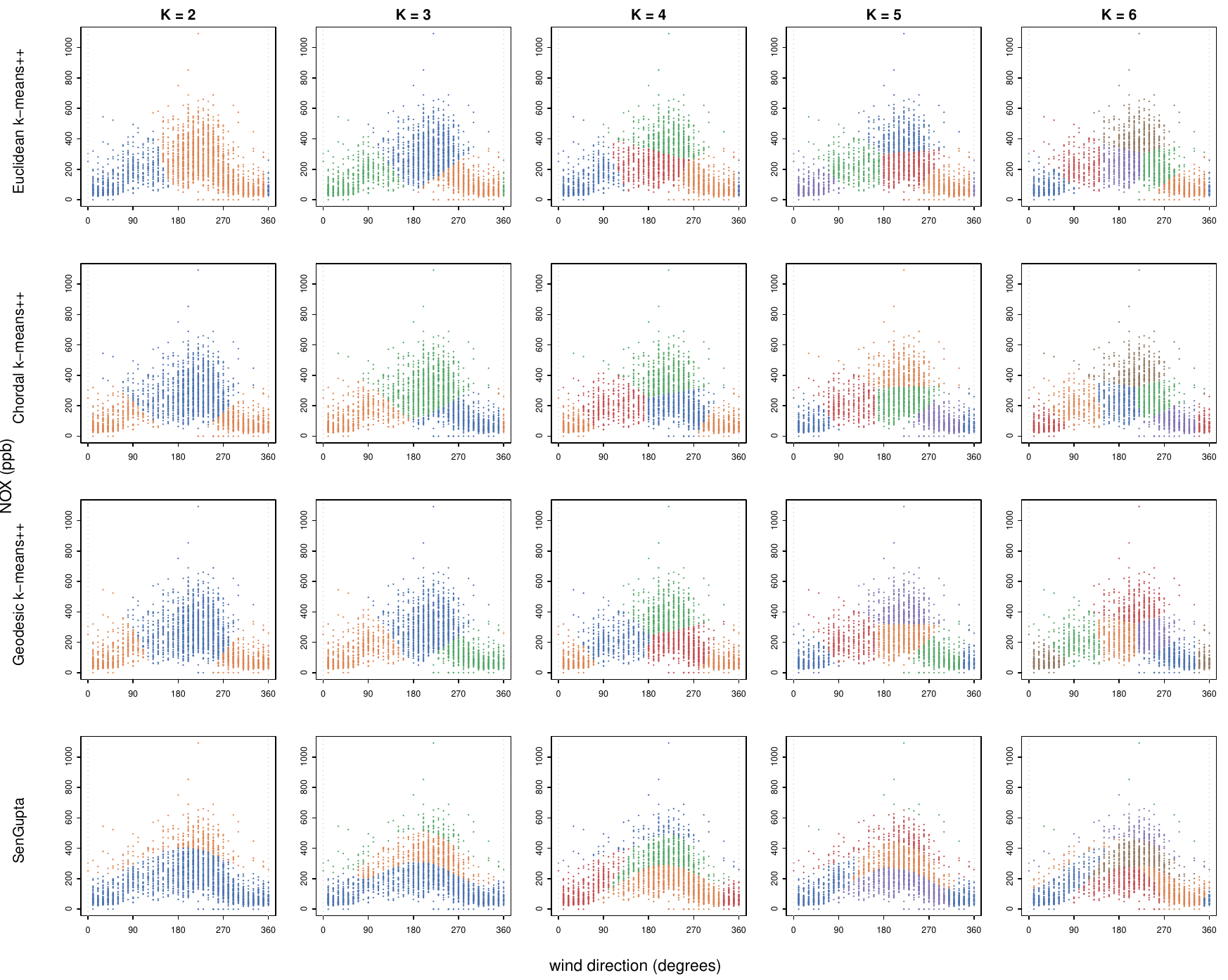}
  \caption{Partitions of the wind direction and \nox\ data for every method (rows) and
    $K=2,\ldots,6$ (columns), in the unrolled view. Colors mark the clusters; the
    dotted lines mark the $0/360^{\circ}$ seam.\label{fig:case1_maps}}
\end{figure}

\end{document}